\documentclass[11pt]{article}
\usepackage{amssymb, amsmath, amsthm, url, times, bbm}
\usepackage[noadjust]{cite}

\usepackage[margin=1in]{geometry}
\usepackage{amsfonts,amsmath,color,float,graphicx,verbatim}

\usepackage{enumerate}
    
\title{Dual Polynomials for Collision and Element Distinctness}
\date{}

\author{Mark Bun\thanks{Harvard University, School of Engineering and Applied Sciences. Supported by an NDSEG Fellowship and NSF grant CNS-1237235.}\\ \texttt{mbun@seas.harvard.edu}
   \and 
Justin Thaler\thanks{Yahoo Labs. Parts of this work were performed while the author was a Research Fellow at the Simons Institute for the Theory of Computing. Supported in part by a Research Fellowship from the Simons Institute for the Theory of Computing.}   \\ \texttt{jthaler@fas.harvard.edu}
}

\usepackage{amsfonts,amsmath,color,float,graphicx,verbatim}

\usepackage{enumerate}

\usepackage{thmtools}
\usepackage{thm-restate}

\newcommand{\EQ}{\mathrm{EQ}}
\newcommand{\NEQ}{\mathrm{NEQ}}

\newcommand{\R}{\mathbb{R}}

\newcommand{\E}{\mathbb{E}}

\newcommand{\bits}{\{-1,1\}}

\newcommand{\eat}[1]{}

\newcommand{\ED}{\operatorname{ED}}

\newcommand{\OR}{\operatorname{OR}}
\newcommand{\MAJ}{\operatorname{MAJ}}

\newcommand{\col}{\operatorname{Col}}
\newcommand{\ed}{\operatorname{ED}}
\newcommand{\ext}{\operatorname{ext}}
\newcommand{\cross}{\mathsf{cross}}

\usepackage{color}
\newcommand{\edit}[1]{{#1}}

\newcommand{\ignore}[1]{}
\newcommand{\provisionallyremove}[1]{}

\newcommand{\eps}{\varepsilon}

\newcommand{\acz}{AC$^{\mbox{0}}$}

\renewcommand{\eqref}[1]{Eq.~(\ref{#1})}

\newcommand{\adeg}{\widetilde{\operatorname{deg}}}

\newtheorem{theorem}{Theorem}
\newtheorem{lemma}[theorem]{Lemma}
\newtheorem{proposition}[theorem]{Proposition}
\newtheorem{corollary}[theorem]{Corollary}

\newtheorem{definition}[theorem]{Definition}

\newtheorem{remark}[theorem]{Remark}

\begin{document}

\maketitle
 \thispagestyle{empty}
\begin{abstract}
The approximate degree of a Boolean function $f: \{-1, 1\}^n \to \{-1, 1\}$ is the minimum degree of a real polynomial that approximates $f$ to within error $1/3$ in the $\ell_\infty$ norm.
In an influential result, Aaronson and Shi (J. ACM 2004) proved tight $\tilde{\Omega}(n^{1/3})$ and $\tilde{\Omega}(n^{2/3})$ lower bounds on the approximate degree of the Collision and Element Distinctness functions, respectively.
Their proof was non-constructive, using a sophisticated symmetrization argument and tools from approximation theory. 

More recently, several open problems in the study of approximate degree have been resolved via the construction of dual polynomials.
These are explicit dual solutions 
to an appropriate
linear program that captures the approximate degree of any function. 
We reprove Aaronson and Shi's results by constructing explicit dual polynomials for the Collision and Element Distinctness functions. 
\end{abstract}

\newpage
\addtocounter{page}{-1}

\section{Introduction} 
The $\eps$-approximate degree of a Boolean function $f:\{-1, 1\}^n \to \{-1, 1\}$ is the least degree of a real polynomial that approximates $f$ to within error $\eps$ in the $\ell_\infty$ norm. Approximate degree is a fundamental measure of the complexity of a Boolean function, and has wide-ranging applications in theoretical computer science. For example, approximate degree upper bounds
underly several of the best known algorithms for PAC learning \cite{ksdnf}, agnostic learning \cite{agnostic, reliable}, learning in the presence of irrelevant information \cite{klivansservedioomb, stt}, and differentially private data release \cite{tuv, ctuw}. 
Meanwhile, lower bounds on approximate degree imply many optimal lower bounds on quantum query complexity, circuit complexity, and communication complexity (see for example \cite{beals, sherstovmajmaj, aaronsonshi, bvdw, shizhu, patternmatrix, beigel, beigelsurvey, sherstovsurvey}).

In an influential result, Aaronson and Shi proved tight $\tilde{\Omega}(n^{1/3})$ and $\tilde{\Omega}(n^{2/3})$ lower bounds on the approximate degree of the Collision and Element Distinctness functions \cite{aaronsonshi}.\footnote{Aaronson established a lower bound of $\tilde{\Omega}(n^{1/5})$ for the Collision function in a paper that appeared in STOC 2002 \cite{aaronson}, and Shi improved it to the tight $\tilde{\Omega}(n^{1/3})$ in a FOCS paper that same year \cite{shi}. A joint journal paper appeared in 2004 \cite{aaronsonshi}. The proof was simplified and extended to the ``small range'' case by Kutin \cite{kutin}. Ambainis  \cite{ambainis} independently extended Aaronson and Shi's lower bound to the small range case, using different techniques than Kutin.} The Collision lower bound matched 
an earlier $O(n^{1/3})$ upper bound due to Brassard et al. \cite{brassard}, while the lower bound for Element Distinctness was later shown to be tight by Ambainis \cite{ambainised}. 

The Collision lower bound subsequently found many applications and extensions in quantum complexity theory; Aaronson recently provided a retrospective overview of these developments \cite{retrospective}. Moreover, the $\tilde{\Omega}(n^{2/3})$ lower bound for Element Distinctness
remains the best known approximate degree lower bound for any function in \acz.

Aaronson and Shi proved their lower bound for Collision with a symmetrization argument. This style of argument proceeds in two steps. First, a polynomial $p$ on $n$ variables (which is assumed to
approximate the target function $f$) is transformed into a polynomial $q$ on $m < n$ variables
in such a way that $\deg(q) \leq \deg(p)$. Second, a lower bound on $\deg(q)$ is proved, typically by applying Markov-Bernstein type inequalities from approximation theory. 
Aaronson and Shi's proof of the Collision lower bound is a particularly sophisticated application of this style of argument.

The lower bound for Element Distinctness follows from a reduction to the lower bound for Collision. This reduction is discussed in Section \ref{sec:ed}.

\medskip
\noindent \textbf{The Method of Dual Polynomials.} 
Despite the many applications of approximate degree in theoretical computer science, significant gaps remain in our understanding of this complexity measure,
and there are many simple functions whose approximate degree
remains unknown. The slow nature of progress can be attributed in part to the limitations of symmetrization arguments. 
At an intuitive level, the process of symmetrization is inherently lossy: by turning a polynomial $p$ on $n$ variables into a polynomial $q$ on $m < n$ variables,
information about $p$ is necessarily thrown away. Hence, several works have identified that an important research direction is to develop techniques beyond symmetrization for lower bounding the approximate degree of Boolean functions \cite{aaronsontutorial, sherstovfocs, bt13}.

 

The last few years have seen significant progress toward this goal. In particular, a series of works has
proved new approximate degree lower bounds for important classes of functions
by constructing explicit \emph{dual polynomials},
which are dual solutions to a certain linear program capturing the approximate degree of any
function. These polynomials act as certificates of the high approximate degree of a function. Moreover,
strong LP duality
implies that the technique is lossless, in contrast to symmetrization. That is, for any function $f$ and any $\eps$, there is always some dual polynomial
$\phi$ that witnesses a tight approximate degree lower bound for $f$; the challenge is to construct $\phi$.

This ``method of dual polynomials'' was recently used to resolve the approximate degree of the AND-OR tree \cite{sherstov13, bt13}, closing a long
line of incrementally larger lower bounds \cite{shi, ambainis, rootn, sherstovfocs, nisanszegedy}. It has also been used to establish several ``hardness amplification''
results for approximate degree \cite{thaler, bt14, sherstov14}, and to prove new \emph{threshold degree} lower bounds for several
important classes of functions, including the intersection of two majorities \cite{odonnellservedio, sherstovfocs} and \acz\ \cite{sherstov14}. The latter result represented
the first superlogarithmic improvement over Minsky and Papert's seminal $\Omega(n^{1/3})$ lower bound from 1969 on the threshold degree of an \acz\ function.
We also note that dual polynomials have recently been used to resolve several longstanding open problems in communication complexity, where they yield explicit
distributions under which various communication problems are hard (see the survey of Sherstov \cite{sherstovsurvey}).

\paragraph{Contribution and motivation.}
We reprove Aaronson and Shi's results by constructing explicit dual polynomials for the Collision and Element Distinctness functions.\footnote{Like Kutin's simplification and refinement
of Aaronson and Shi's original proof of the Collision lower bound, our construction yields a dual polynomial for the Collision function
even in the ``small-range'' case.} First, we give a direct construction of a dual polynomial for Collision. In Section \ref{sec:overview}, we give an overview of the ideas that go into this construction. We then show how to turn any dual polynomial $\psi$ for Collision into a dual polynomial $\varphi$ for Element Distinctness. We construct $\varphi(x)$ by averaging $\psi(y)$ over a carefully constructed set of extensions from each $x$ to a longer input $y$.

We have two main motivations for reproving Aaronson and Shi's lower bound in this manner.
First, only a handful of techniques are currently known for the construction of dual polynomials, especially
for the case where $\eps=\Theta(1)$. To date, dual polynomials have been constructed only
for symmetric functions \cite{spalek, bt13} and a handful of highly structured block-composed functions  \cite{bt14, bt13, sherstov13, sherstovfocs, sherstov14, comm5} (a block-composed
function $F \colon \{-1, 1\}^{M \cdot N} \rightarrow \{-1, 1\}$ is a function of the form
of the form $F=g(f(x_1), \dots, f(x_M))$ for some $g \colon \{-1, 1\}^M \rightarrow \{-1, 1\}$
and $f \colon \{-1, 1\}^N \rightarrow \{-1, 1\}$). The Collision and Element
Distinctness functions fall into neither category; 
our constructions of dual polynomials for these problems introduce several new techniques that we hope will prove useful
in future applications. 

A second motivation is to shed new light on the Collision lower bound itself. 
The earlier symmetrization-based proof \cite{aaronsonshi, kutin}, while shorter than ours, is non-constructive and relies on Markov-Bernstein inequalities from approximation theory.
In contrast, our proof is constructive and entirely
elementary.
We also believe that our analysis illuminates some of the more miraculous aspects of the earlier symmetrization-based proof -- see Section \ref{sec:discussion} for further
discussion of this point.

\paragraph{Related work on quantum query complexity.} 
Aaronson and Shi's original motivation for studying the approximate degree of the Collision function was to understand its quantum query complexity (recall that approximate degree provides 
a \emph{lower bound} on quantum query complexity \cite{beals}. However, it is known that the lower bound is not always tight \cite{qqcpolyseparation}). Subsequent to Aaronson and Shi's work, other methods were developed for quantum query complexity \cite{adversary1, adversary2, adversary3, qqcpolyseparation, adversary4, adversary5},
and it is now known that one of these methods, called the \emph{negative-weights adversary method} \cite{adversary1}, is always tight. 

The negative-weights adversary method for lower bounding quantum query complexity is closely analogous to the method of dual polynomials for approximate degree: 
the former is characterized by a semidefinite program, and a solution to this semidefinite program is known as an \emph{adversary matrix}. 
A recent line of work, similar in spirit to our own, has proved or reproved optimal quantum query complexity lower bounds for several functions by constructing explicit adversary matrices.
In particular, Belovs and Rosmanis \cite{belovsrosmanis} constructed an optimal adversary matrix for the Collision function in the ``large range'' case (note that the
dual polynomial that we construct applies even in the ``small range'' case), and Belovs and \v{S}palek constructed an optimal adversary matrix for the Element Distinctness function \cite{belovsspalek}. 

Very recently, Zhandry \cite{zhandry} (improving on work of Yuen \cite{yuen}) has also proved a tight lower bound of $\Omega(N^{1/3})$ on the quantum query complexity of finding a collision in a randomly chosen function.






\section{Preliminaries}

\subsection{Notation}
For any positive integer $n$, we denote the set $\{1, \dots, n\}$ by $[n]$, and the set $\{0, 1, \dots, n\}$ by $[[n]]$. \edit{For a function $f : D \to \R$, define the $L_1$ norm $\|f\|_1 = \sum_{x \in D} |f(x)|$.} For any subset $S \subseteq [n]$,
we let $\chi_S \colon \{-1, 1\}^n \rightarrow \{-1, 1\}$ denote the parity function on $S$, i.e., $\chi_S(x) = \prod_{i \in S} x_i$.

\subsection{Approximate Degree and its Dual Characterization}
\label{sec:dualdefs} \label{sec:prelims}
Let $D \subseteq \{-1, 1\}^n$, and let $f \colon D \rightarrow \{-1, 1\}$ be a partial Boolean function
defined on $D$. A real polynomial $p \colon \{-1, 1\}^n \rightarrow \{-1, 1\}$ is said to $\eps$-approximate $f$
if 

\begin{enumerate}
\item $|p(x) - f(x)| \leq \eps$ for all $x \in D$, and
\item $|p(x)| \le 1 + \eps$ for all $x \in \{-1, 1\}^n$.
\end{enumerate}

The $\eps$-approximate degree of $f$, denoted $\adeg_{\eps}(f)$, is the minimum degree of an $\eps$-approximation for $f$. 
We use $\adeg(f)$ to denote $\adeg_{1/3}(f)$, and refer to this quantity without qualification as the \emph{approximate degree} of $f$. The choice of $1/3$ is arbitrary, as $\widetilde{\deg}(f)$ is related to $\adeg_\eps(f)$ by a constant factor for any constant $\eps \in (0, 1)$. 

Given a partial Boolean function $f$, let $p$ be a real polynomial that attains the smallest $\eps$ subject to the constraints above, over
all polynomials of degree at most $d$. Since we work over $x \in \{-1, 1\}^n$, we may assume without loss of generality that $p$ is multilinear with the representation $p(x) = \sum_{|S| \leq d} c_S \chi_S(x)$, where the coefficients $c_S$ are real numbers. Then $p$ is an optimum of the following linear program.

\[ \boxed{\begin{array}{lll} 
    \text{min}  &     \eps    \\
    \mbox{such that} &\Big|f(x) - \sum_{|S| \leq d} c_S \chi_S(x)\Big| \leq \eps & \text{ for each } x \in D\\
    &\Big|\sum_{|S| \leq d} c_S \chi_S(x)\Big| \leq 1 + \eps & \text{ for each } x \in \{-1, 1\}^n \setminus D\\
    &c_S \in \mathbb{R} & \text{ for each } |S| \leq d\\
    &\eps \geq 0
    \end{array}}
\]

The dual linear program is as follows.

\[ \boxed{\begin{array}{lll} 
    \text{max} &    \sum_{x \in D} \phi(x) f(x)  - \sum_{x \in \{-1, 1\}^n \setminus D} |\phi(x)| \\
    \mbox{such that} &\sum_{x \in \{-1, 1\}^n} |\phi(x)| = 1\\
    &\sum_{x \in \{-1, 1\}^n} \phi(x) \chi_S(x)=0  & \text{ for each } |S| \leq d\\
    &\phi(x) \in \mathbb{R} & \text{ for each } x \in \{-1, 1\}^n
    \end{array}}
\]

Strong LP-duality thus implies the following dual characterization of approximate degree:

\begin{theorem} \label{thm:prelim} Let $f: D \to \{-1, 1\}$ be a partial Boolean function. Then $\adeg_\eps(f) > d$ if and only if there is a polynomial $\phi \colon \{-1, 1\}^n \rightarrow \mathbb{R}$ such that 
\begin{equation} \label{eq:prelim0} \sum_{x \in D} f(x) \phi(x) - \sum_{x \in \{-1, 1\}^n \setminus D} |\phi(x)| > \eps \cdot \sum_{x \in \bits^n} |\phi(x)|, \end{equation}
and
\begin{equation} \label{eq:prelim2} \sum_{x \in \{-1, 1\}^n} \phi(x) \chi_S(x)=0   \text{ for each } |S| \leq d.\end{equation}
\end{theorem}

If $\phi$ satisfies \eqref{eq:prelim0}, we say that $\phi$ has \emph{correlation greater than $\eps$} with $f$.
If $\phi$ satisfies \eqref{eq:prelim2}, we say $\phi$ has \emph{pure high degree} $d$.
We refer to any feasible solution $\phi$ to the dual linear program as an $(\eps, d)$-\emph{dual polynomial} for $f$.

\subsection{The Collision and Element Distinctness Functions}

Let $[N]=\{1, \dots, N\}$, and fix a triple of positive integers $n, N, R$ such that $R \geq N$, and
$n = N \cdot \log_2 R$. For simplicity throughout, we assume that $R$ is a power of 2. The Collision and Element Distinctness functions are typically thought of as 
\emph{properties} of functions mapping $[N]$ to $[R]$. However, it will be convenient for us to think of them instead
as functions on the Boolean hypercube $\{-1, 1\}^n$. 
To this end, given an input $x \in \{-1, 1\}^n$, we interpret $x$ as the evaluations of a function $g_x$ mapping $[N] \rightarrow [R]$. That is, we break $x$ up into $N$ blocks, each of length $\lceil \log_2 R \rceil$, and regard each block $x_i$ as the binary representation of $g_x(i)$. 

\begin{definition}[Collision Function]
A function $g_x \colon [N] \rightarrow [R]$ is said to be $k$-to-1 if for every $i \in [N]$, there exists exactly $k-1$ values $j \neq i$ such that $g_x(i) = g_x(j)$. 
Let $T_k := \{x \in \{-1, 1\}^n : g_x \text{ is } k\text{-to-}1\}$ (clearly, $T_k$ is non-empty only if $k | N$).
The Collision function,  which we denote by $\col_{N, R}$, is the partial Boolean function defined on $T_1 \cup T_2 \subseteq \{-1, 1\}^n$
such that $\col_{N, R}(x)=1$ if and only if $x \in T_1$. That is, $\col_{N, R}$ is the partial Boolean function corresponding to the property that $g_x$ is a 1-to-1 function, with the promise that $g_x$ is either 1-to-1 or 2-to-1.
\end{definition}

\begin{definition}[Element Distinctness Function]
The \emph{Element Distinctness function}, denoted $\ed_{N, R}$, is the total Boolean function defined such that $\ed_{N,R}(x)=1$ if and only if $g_x$ is 1-to-1.
That is,  $\ed_{N,R}$ is the total Boolean function corresponding to the property that $g_x$ is 1-to-1.
\end{definition}

Let $B \subset \{-1, 1\}^n$ denote the set of inputs $x$ such that $g_x$ is neither 1-to-1 nor 2-to-1. Then 
an $(\eps, d)$-dual polynomial $\phi$ for $\col_{N, R}$ has the following properties (cf. Section \ref{sec:dualdefs}):

\begin{enumerate}
\item $\sum_{x \in T_1} \phi(x) - \sum_{x \in T_2} \phi(x) - \sum_{x \in B} |\phi(x)| > \eps \cdot \sum_{x \in \bits^n} |\phi(x)|$.
\item $\sum_{x \in \bits^n} \phi(x) \chi_S(x) = 0 \mbox{ for all } |S|\leq d$.
\end{enumerate}

Similarly, an $(\eps, d)$-dual polynomial for $\ed_{N,R}$ satisfies:

\begin{enumerate}
\item $\sum_{x \in T_1} \phi(x) - \sum_{x \notin T_1} \phi(x) > \eps \cdot \sum_{x \in \bits^n} |\phi(x)|$.
\item $\sum_{x \in \bits^n} \phi(x) \chi_S(x) = 0 \mbox{ for all } |S|\leq d$.
\end{enumerate}

\subsection{Overview of the Symmetrization-Based Proof of the Collision Lower Bound}
\label{sec:kutin}
Kutin's simplified proof of the Collision lower bound \cite{kutin} proceeds in two steps. The first step is a symmetrization step, which
establishes the following remarkable result (we state this result slightly informally in this overview).

\begin{lemma}[Informal version of Lemma \ref{lem:symmetrization}] \label{lem:informal} Call a triple $(m, a, b)$ \emph{valid} if $a | m$ and $b|(N-m)$. 
For any
triple $(m, a, b)$, let $R_{m, a, b}$ denote the set of inputs $x \in \{-1, 1\}^n$ such that
$g_x \colon [N] \rightarrow [R]$ maps $m$ of its inputs to $[R]$ in an $a$-to-1 manner, and maps the remaining $N-m$ of its inputs to $[R]$
in a $b$-to-1 manner. Then there is a trivariate polynomial $P$ of total degree at most $d$ such that for every valid triple 
$(m, a, b)$, it holds that $P(m, a, b)=\E_{x \in R_{m, a, b}}[p(x)]$.
\end{lemma}

Note in the above lemma that the sets $R_{m, a, b}$ are not uniquely determined; for instance $R_{m, 1, 1} = R_{0, a, 1} = R_{N, 1, b} = T_1$ for every triple $(m, a, b)$.

The second step of Kutin's proof argues that if $p$ is a $(1/3)$-approximating polynomial for the Collision function, then $P$ 
must have degree $\Omega(N^{1/3})$. Hence by Lemma \ref{lem:informal}, $p$ must have degree $\Omega(N^{1/3})$ as well. 

In more detail, the second step of Kutin's proof proceeds via a case analysis. Four cases are considered. 
\begin{itemize}
\item
The first is: 
$P(N/2, 1, 2) \geq 1/2$, and $|P(N/2, 1, b)| \leq 2$  for all $b \in [N^{2/3}]$. In this case, Kutin is able to apply Markov's inequality
from approximation theory to conclude that the degree of $P$ in its third variable is $\Omega(N^{1/3})$.

\item The second is: 
$P(N/2, 1, 2) \geq 1/2$, and $|P(N/2, 1, b)| > 2$  for some $b \in [N^{2/3}]$. In this case, Kutin is able to apply Bernstein's inequality
from approximation theory to conclude that the degree of $P$ in its first variable is $\Omega(N^{1/3})$.

\item The third is: 
$P(N/2, 1, 2) < 1/2$, and $|P(N/2, a, 2)| \leq 2$  for all $a \in [N^{2/3}]$. In this case, Kutin is able to apply Markov's inequality
 to conclude that the degree of $P$ in its second variable is $\Omega(N^{1/3})$.

\item The fourth is: 
$P(N/2, 1, 2) < 1/2$, and $|P(N/2, a, 2)| > 2$  for some $a \in [N^{2/3}]$. In this case, Kutin is able to apply Bernstein's inequality
 to conclude that the degree of $P$ in its first variable is $\Omega(N^{1/3})$.
\end{itemize}
A key technical complication that must be 
dealt with in the argument above is that $|P(m, a, b)|$ may be much larger than 1 for \emph{invalid} triples $(m, a, b)$. 
This may seem like a minor technicality, but in fact it is a central issue: if $P(m, a, b)$ were bounded for all invalid triples,
then it would be possible to argue that the total degree of $P$ is $\Omega(N^{1/2})$, which would imply a (false) lower bound of $\Omega(N^{1/2})$
on the approximate degree of $\col_{N, R}$.

\subsection{Overview of Our Construction for the Collision Function}
\label{sec:overview}
Like Kutin's proof, our construction also makes essential use of Lemma \ref{lem:informal}. Whereas Kutin used Lemma \ref{lem:informal}
to reduce to a setting where Markov-Bernstein inequalities could be applied in a non-constructive manner, we instead use Lemma \ref{lem:informal}
to argue that the dual polynomial $\phi$ that we construct has pure high degree $\Omega(N^{1/3})$. 

In more detail, we present our construction in two stages, in order to highlight distinct ideas that go into the proof. 
In the first stage, we construct a simpler dual polynomial $\phi \colon \{-1, 1\}^n \rightarrow \{-1, 1\}$ that exhibits an $\Omega(\sqrt{\log N/\log\log N})$ 
lower bound on the approximate degree of $\col_{N, R}$. The second stage constructs a dual polynomial $\psi$ exhibiting the optimal
$\Omega(N^{1/3})$ lower bound.

\paragraph{Overview of the first stage.} 
Let $H_k\subseteq \{-1, 1\}^n$ denote the set of inputs of Hamming weight $n$. The symmetrization-based proof of the Collision lower bound from \cite{aaronsonshi, kutin} carries the strong intuition that
the sets $T_k$ should play the same role that $H_k$ plays in Nisan and Szegedy's seminal symmetrization-based lower bound for the $\OR$ function \cite{nisanszegedy}. We direct
the interested reader to Aaronson's lecture notes \cite{lecturenotes} for a detailed explanation of this intuition. The construction of our simpler dual witness $\phi$ instantiates this intuition in the dual setting. 

Recall that a dual polynomial $\phi$ witnessing the fact that $\adeg(\col_{N, R}) \geq d$ must satisfy two properties: (1) it must
have correlation greater than $\eps$ with $\col_{N, R}$, and (2) it must have pure high degree at least $d$.
We define $\phi$ in a way that mimics the structure of known dual witnesses for symmetric functions, even though $\phi$ is not itself symmetric. 
Specifically, our construction ensures that the analysis establishing Properties (1) and (2) becomes similar to the analyses of known dual polynomials for
the $\OR$ function \cite{spalek, bt13}.

In more detail, our prior work \cite{bt13} built on work of \v{S}palek \cite{spalek} to give a dual witness $\gamma$ for the fact that $\adeg_{\eps}(\OR_n)=\Omega(\sqrt{n})$ for any constant $\eps < 1$; moreover,
$\gamma$ places non-zero weight only on sets $H_k$, for values of $k$ equal (up to scaling factors)
to perfect squares. The pure high degree of $\gamma$ is shown to be equal to (at least) the number of sets $H_k$ upon which $\gamma$ places non-zero weight. 

Call an input $x \in \{-1, 1\}^n$ \emph{valid} if it is in $R_{m, a, b}$ for some valid triple $(m, a, b)$.
By analogy with $\gamma$, the dual witness $\phi$ that we construct in Stage 1 places weight only on inputs $x \in T_k$ for divisors $k$ of $N$ that are also (up to scaling factors) perfect squares.
In particular, our definition of $\phi$ ensures that:
\begin{equation} \phi(x)=0 \text{ for all invalid inputs } x. \label{whoop} \end{equation} 

We are able to combine \eqref{whoop} with Lemma \ref{lem:informal} and a basic combinatorial identity (cf. Lemma \ref{lem:combinatorial}) to show that the pure high degree of $\phi$ is at least $|S|$,
where $S$ denotes the set of $T_k$'s upon which $\phi$ places non-zero weight.
Moreover, our definition of $\phi$ is carefully chosen to ensure that its correlation with $\col_{N, R}$ is large: the precise calculation is closely analogous to 
the analysis from \cite{spalek, bt13} showing that $\gamma$ is well-correlated with the $\OR$ function \cite{spalek, bt13}.

\paragraph{Overview of the second stage.}  In the second stage, we construct a dual polynomial $\psi$ that exhibits the optimal $\Omega(N^{1/3})$ lower bound. Rather than only weighting inputs in $T_k$ for some some divisors $k$ of $n$, $\psi$ weights inputs in $R_{m, a, b}$ for many valid triples $(m, a, b)$. There are two key ideas that go into the construction of $\psi$. 

The first idea is to define $\psi$ as the sum of two simpler dual polynomials $\psi_1$ and $\psi_2$, each with pure high degree $\Omega(N^{1/3})$ -- then the sum $\psi$ also has pure high degree $\Omega(N^{1/3})$ (see Lemma \ref{lem:phd-sum}). The first polynomial $\psi_1$ places a large constant fraction (close to $1/2$) 
of its $L_1$ mass on $T_1$, whereas $\psi_2$ places a large constant fraction of its $L_1$ mass on $T_2$. 
Neither $\psi_1$ nor $\psi_2$ is well-correlated with 
$\col_{N, R}$ in the sense of \eqref{whoop}. However, they each place a constant fraction of their $L_1$ mass on $R_{N/2, 2, 1}$,
and they are designed so that their values exactly cancel out on inputs in $R_{N/2, 2, 1}$. 
This allows us to show that $\psi=\psi_1 + \psi_2$ 
satisfies \eqref{whoop}, even though $\psi_1$ and $\psi_2$ individually do not.

The second idea goes into the construction of $\psi_1$ and $\psi_2$ themselves. Specifically, we think of $\psi_1$ and $\psi_2$ as each being
constructed in a two-step process. We focus on $\psi_1$ in this discussion, since the construction of $\psi_2$ is similar. 
Very roughly speaking, in the first step, we consider a ``polynomial'' $\psi'$ of pure high degree $\Omega(N^{1/3})$ that places a large constant fraction of its $L_1$ mass
on $T_1$; the construction of $\psi'$ is closely related to our construction of the simpler dual polynomial $\phi$ from Stage 1.

The reason we place the term ``polynomial'' in quotes above is that there is an important technical caveat to our construction of $\psi'$: we think of $\psi'$ as 
placing weight on sets $R_{N/2, a, 1}$ for many \emph{invalid} triples $(N/2, a, 1)$, in addition to some valid ones. Of course, if $(N/2, a, 1)$ is invalid, then $R_{N/2, a, 1}=\emptyset$,
so $\psi'$ cannot place non-zero weight on the set. 
To address this issue, in Step 2, we add to $\psi'$ a bunch of polynomials $\psi''_{N/2, a, 1}$, each of pure high degree $\Omega(N^{1/3})$. For each invalid triple $(N/2, a, 1)$,
$\psi''_{N/2, a, 1}$ is specifically
constructed to cancel out the weight that $\psi'$ ``places'' on $R_{N/2, a, 1}$. 

Analogously to how our constructions of $\phi$ and $\psi'$ were closely
related to the dual witness for $\OR$ constructed in our earlier work \cite{bt13}, our construction of $\psi''_{N/2, a, 1}$ is closely related to a 
dual witness $\eta$ for the Majority function, $\MAJ$, that we constructed in the same work.
Each $\psi''_{N/2, a, 1}$ places additional non-zero mass on (non-empty) sets of the form $R_{m, a, 1}$ for some $a\neq 1$ and $m \in [N]$, but
we are able to show that the total mass placed on such sets is small, using an analysis closely related to the analysis of $\eta$ from \cite{bt13}.
Hence we are able to show that 
$\psi_1 = \psi' + \sum_{\text{invalid triples } (N/2, a, 1)} \psi''_{N/2, a, 1}$ still places a large constant fraction of its $L_1$ mass on $T_1$. 

\subsection{Discussion}
\label{sec:discussion}
\paragraph{On Kutin's second step.} Our construction of the optimal dual witness $\psi$ for the Collision function mimics the second step of Kutin's symmetrization argument in three important ways described below. We find this mimicry to be somewhat surprising -- in our earlier work \cite{bt13}, we constructed an optimal dual polynomial for symmetric Boolean functions
that bore little relation to Paturi's well-known symmetrization-based proof of the same result \cite{paturi}. 
We believe that this mimicry sheds new light, or at least gives a new perspective, on why Kutin's proof takes the structure that it does.

Recall that the second step of Kutin's proof (cf. Section \ref{sec:kutin}) proceeds via a case analysis. The first ``branch''
in the case analysis depends on whether
the expected value of the assumed $n$-variate approximation $p$ to $\col_{N, R}$ on the set $R_{N/2, 2, 1}$ is large or small. 
This is mimicked in our construction of $\psi$ as a sum of two dual polynomials $\psi_1$ and $\psi_2$, both of which individually 
place a lot of weight on $R_{N/2, 2, 1}$, but whose sum places \emph{zero} weight on $R_{N/2, 2, 1}$. 

The second ``branch'' in Kutin's case analysis depends on whether $|P(N/2, a, 1)|$ or $|P(N/2, 2, b)|$ is small for all $a, b \leq N^{2/3}$. 
He needs to consider this second branch because $P(m, a, b)$ is not guaranteed to be bounded for invalid triples $(m, a, b)$.
 
This branch is mimicked in our construction of $\psi_1$ (respectively, $\psi_2$) as the sum of a single ``polynomial'' $\psi'$ that tries to
place weight on sets $R_{N/2, a, 1}$ for invalid triples $(N/2, a, 1)$ (respectively, $(N/2, 2, b)$), and many other polynomials $\psi''_{N/2, a, 1}$ (respectively, $\psi''_{N/2, 2, b}$), 
one for each invalid triple 
$(N/2, a, 1)$ (respectively, $(N/2, 2, b)$). In our dual setting, the reason we need to incorporate the polynomials $\psi''_{N/2, a, 1}$ is to cancel out the weight that $\psi'$ tries to place
on invalid sets $R_{N/2, a, 1}$. 

Finally, recall that Kutin applied Markov's inequality from approximation theory in two of the four cases considered in his analysis, and Bernstein's inequality in the other two cases.
Markov's inequality underlies Nisan and Szegedy's standard symmetrization-based proof that the approximate degree of $\OR$ is $\Omega(\sqrt{n})$ \cite{nisanszegedy}, while Berstein's
inequality underlies Paturi's proof that the approximate degree of $\MAJ$ is $\Omega(n)$ \cite{paturi}. This is mimicked in our construction of $\psi_1$ and $\psi_2$ as the sum of $\psi'$
and the $\psi''_{N/2, a, 1}$ and $\psi''_{N/2, 2, b}$ polynomials: the construction of $\psi'$ is closely analogous to the dual witness for $\OR$  from \cite{bt13},
while the construction of the $\psi''_{N/2, a, 1}$ and $\psi''_{N/2, 2, b}$ polynomials is based on the dual witness for $\MAJ$ from \cite{bt13}.

\paragraph{On the first step, or why $k$-to-1 inputs matter.} As noted by several authors (e.g., \cite[Slide 36]{aaronsontutorial}), the most miraculous element of the symmetrization-based proof
of the Collision lower bound is the first step (cf. Lemma \ref{lem:informal}). 
The crux of this step is to establish, roughly speaking, that for any $n$-variate polynomial $p$ of total degree $d$, 
the function $P(k) := \E_{x \in T_k}[p(x)]$ is a polynomial in $k$ of degree at most $d$. Why should this hold? More basically, why should inputs
that are $k$-to-$1$ even play a prominent role in the proof? 

We provide some partial intuition for this in Section \ref{sec:complementaryslackness}. Specifically, we explain that there is an (asymptotically) optimal 
approximation $p$ for $\col_{N, R}$ such that $k$-to-1 inputs correspond to constraints that are made tight
by the solution corresponding to $p$ in the primal linear program of Section \ref{sec:prelims}. Hence, complementary slackness suggests that there should be a corresponding dual
witness $\psi$ that places weight only on inputs that are $k$-to-1, or nearly so, justifying the prominent role that $k$-to-1 inputs play in 
both the symmetrization-based proof and our new dual proof.




\subsection{Formal Statement of Lemma \ref{lem:informal}}
Following Kutin \cite{kutin}, we define a special collection of functions which are $a$-to-1 on one part of the domain and $b$-to-1 on the other part. For $N > 0$, recall that a triple of numbers $(m, a, b)$ is \emph{valid} if $a | m$ and $b | (N-m)$. For each valid triple $(m, a, b)$, we define
\[g_{m, a, b}(i) = \begin{cases}
\lceil i / a \rceil \quad \text{if } 1 \le i \le m \\
R - \lfloor (N - i) / b \rfloor \quad \text{if } m < i \le n.
\end{cases}\]

Moreover, for each valid triple $(m, a, b)$, we define a set $R_{m, a, b}$ that is the orbit of $g_{m, a, b}$ under the automorphism group $S_N \times S_R$. Namely,
\[x \in R_{m, a, b} \iff \exists \sigma \in S_N, \tau \in S_R \colon \quad \tau \circ g_x \circ \sigma = g_{m, a, b}.\]
Note that the sets $R_{m, a, b}$ are not uniquely determined; for instance $R_{m, 1, 1} = R_{0, a, 1} = R_{N, 1, b} = T_1$ for every $m, a, b$.

\begin{lemma} \label{lem:symmetrization}
Let $p(x)$ be a real polynomial over $\bits^n$ of degree $d$. There is a trivariate polynomial $P$ of degree at most $d$ with
the property that for all valid triples $(m, a, b)$,
\[P(m, a, b) = \E_{x \in R_{m, a, b}}[p(x)].\]
\end{lemma}

The statement of 
Lemma \ref{lem:symmetrization} differs slightly from the corresponding lemma in Kutin's work \cite{kutin} (Lemma \ref{lem:kutin} below).
Lemma \ref{lem:symmetrization} follows by combining Kutin's formulation with the following
simple lemma from \cite{bt14}.

\begin{lemma}[\cite{bt14}] \label{lem:representation}
Let $p$ be a polynomial over $\{-1, 1\}^n$. Consider the map $T \colon \{-1, 1\}^m \to \{0, 1\}^{N\cdot R}$ defined by $T_{ij}(x) = 1$ if $g_x(i) = j$, and $T_{ij}(x) = 0$ otherwise. Then there is a polynomial $q \colon \{0, 1\}^{N \cdot R} \to \R$ with $\deg q \le \deg p$, such that $q(T(x)) = p(x)$ for all $x \in \{-1, 1\}^n$.
\end{lemma}

%

\begin{lemma}[\cite{kutin}]\label{lem:kutin}
Let $q(t)$ be any degree $d$ polynomial in the variables $t_{ij}$. For a valid triple $(m, a, b)$, define $Q(m, a, b)$ by
\[P(m, a, b) = \E_{x \in R_{m, a, b}}[q(T(x))].\]
Then $P$ is a degree $d$ polynomial in $m, a, b$.
\end{lemma}

\section{An $\Omega(\sqrt{\log N / \log \log N})$ Lower Bound for the Collision Function}
\label{sec:firstlb}
The following lemma is a refinement of \cite[Proposition 14]{bt13}, which was used there to construct a dual polynomial for $\OR$.

\begin{lemma} \label{lem:or-dual}
There exists a constant $\zeta > 0$ such that for all $\delta \in (0, 1)$ and $L \ge 1$, there is an explicit $\omega: \{1, \dots, L\} \to \R$ with
\begin{enumerate}
\item $\omega(1) \ge \frac{1-\delta}{2}$
\item $- \omega(2) \ge \frac{1-\delta}{2}$
\item $\sum_{k=1}^L |\omega(k)| = 1$
\item For every polynomial $p: \{1, \dots, L\} \to \R$ of degree $d \le \zeta\sqrt{\delta L}$, we have $\sum_{k = 1}^L  p(k) \omega(k) = 0$.
\end{enumerate}
\end{lemma}

The proof will make use of the following simple combinatorial identity, a simple proof of which can be found in 
\cite[Appendix A]{odonnellservedio}.

\begin{lemma} \label{lem:combinatorial}
For any $L > 0$, let $q \colon \mathbb{R} \rightarrow \mathbb{R}$ be a univariate polynomial of degree strictly less than $L$.
Then $\sum_{k=0}^{L} (-1)^k {L \choose k} q(k) = 0$.
\end{lemma}

\begin{proof}[Proof of Lemma \ref{lem:or-dual}]
Let $c = \lceil 16 / \delta \rceil$. Let $m = \lfloor \sqrt{(L-1)/c} \rfloor$ and define the set
\[T = \{1\} \cup \{ci^2 \colon 0 \le i \le m\}.\]
Note that $|T| = \Omega(\sqrt{L/c})$. Define the function $\hat{\omega}: \{0, 1, \dots, L\} \to \R$ by
\[\hat{\omega}(k) = {L \choose k} \frac{c^{m}(m!)^2}{L!} \prod_{j \in [[L]] \setminus T} (k - j) = 
\begin{cases}
\dfrac{c^m(m!)^2}{\prod_{j \in T \setminus \{k\}} (k - j)} \quad \text{ if } k \in T, \\
\\
0 \quad \text{ otherwise.}
\end{cases}\]
It is easy to check that $\hat{\omega}(0) = 1$.

For $k = 1$, we have $|\hat{\omega}(k)|=\frac{c^m(m!)^2}{\prod_{i = 1}^m (ci^2 - 1)}$. Notice that:
\begin{align*}
1 \le \frac{c^m(m!)^2}{\prod_{i = 1}^m (ci^2 - 1)} &= \prod_{i=1}^m \frac{i^2}{i^2 - 1/c} \\
&= \prod_{i=1}^m \left(1 + \frac{1}{ci^2 - 1}\right) \\
&\le \exp \left(\sum_{i=1}^m \frac{2}{ci^2}\right) \\
&\le e^{8/c} \le 1 + \frac{3}{2}\delta.
\end{align*}

On the other hand, for $k = c\ell^2$ with $\ell > 0$, $|\hat{\omega}(k)|$ equals:
\begin{align*}
\frac{c^m(m!)^2}{(c\ell^2 - 1)\prod_{i \in \edit{[[m]]} \setminus \{\ell\}} |c\ell^2 - ci^2|} &= \frac{(m!)^2}{(c\ell^2 - 1)\prod_{i \in \edit{[[m]]} \setminus \{\ell\}} (i + \ell)|i - \ell|} \\
&= \frac{2(m!)^2}{(c\ell^2 - 1)(m+\ell)!(m-\ell)!} \\
&\le \frac{2}{c\ell^2 - 1},
\end{align*}
where the last inequality follows because
\[\frac{(m!)^2}{(m+\ell)!(m-\ell)!} = \frac{m}{m+\ell} \cdot \frac{m-1}{m+\ell-1} \cdot \ldots \cdot \frac{m-\ell+1}{m+1}\]
is a product of factors that are each smaller than 1. Thus, the total contribution of terms excluding $0$ and $1$ to the $\ell_1$ mass of $\hat{\omega}$ is at most
\[\sum_{i=1}^{m} \frac{2}{ci^2 - 1} < \sum_{i=1}^\infty \frac{4}{ci^2} < \frac{8}{c} \le \frac{\delta}{2}.\]

\edit{Now let $b=0$ if $|T|$ is even, and $b=1$ otherwise}, and define $\omega \colon \{1, \dots, L\} \rightarrow \R$ via:
$$\omega(k) = (-1)^{k+b}\hat{\omega}(k-1) / \|\hat{\omega}\|_1.$$ Then
\[-\omega(2) \ge \omega(1) \ge \frac{1}{1 + |\hat{\omega}(1)| + \edit{\delta/2}} \ge \frac{1}{2 + 2\delta} \ge \frac{1-\delta}{2}.\]
This yields the first two claims about $\omega$. The third claim follows immediately from the definition. Finally, let $p$ be a polynomial of degree strictly less than $|T|-1$. Then
\begin{equation} \sum_{k=1}^{L} p(k) \omega(k) = \sum_{k=0}^{L-1} (-1)^k \cdot (-1)^b \cdot {L \choose k} \cdot \frac{c^{m}(m!)^2}{L! \|\hat{\omega}\|_1} \cdot p(k+1)\prod_{j \in [[L]] \setminus T} (k - j) = \sum_{k=0}^{L - 1} (-1)^k {L \choose k} q(k), \label{eq:finale} \end{equation}
where
\[q(k) = \frac{(-1)^b c^{m}(m!)^2}{L! \|\hat{\omega}\|_1} p(k+1)\prod_{j \in \edit{[[L]]} \setminus T} (k - j)\]
is a polynomial of degree less than $L$. Since $q(L) = 0$, the right hand side of \eqref{eq:finale} is zero by Lemma \ref{lem:combinatorial}. This gives the last claim.
\end{proof}

Our prior work \cite{bt13}, building on work of \v{S}palek \cite{spalek}, obtained a dual polynomial $\gamma$ for $\OR_L$ by setting the total weight of $\gamma$ on inputs in $H_k$ (the set of inputs of Hamming weight $k$) to be $\omega(k + 1)$. In that work, the first three properties of $\omega$ ensured that $\gamma$ had high correlation with $\OR$, while the fourth ensured that it had pure high degree $\Omega(\sqrt{L})$.

Analogously, our dual polynomial $\phi$ for $\col_{N, R}$ below sets the total weight of $\phi$ on $T_k$ to be $\omega(k)$. Then again, the first three properties of $\omega$ ensure that $\phi$ is well-correlated with $\col_{N, R}$, and the fourth ensures that it has pure high degree $\Omega(\sqrt{L})$. However, there is the complication that $T_k$ must be non-empty, i.e., $k$ must divide $N$, for every $k$ in the support of $\omega$. To handle this complication, we take $N$ large enough so that all $k = 1, 2, \dots, L$ divide $N$, yielding an $\Omega(\sqrt{\log N/\log\log N})$ lower bound.

\begin{theorem}
Let $N = L!$ for some $L$. For $\delta > 0$, there exists an explicit $(1-2\delta, d)$ dual polynomial $\phi$ for $\col_{N, R}$ with $d = \Omega(\sqrt{\delta L}) = \Omega(\sqrt{\delta \log N/\log\log N})$.
\end{theorem}

\begin{proof}
First, notice that $k | N$ for all $k \in [L]$, so $T_k \neq \emptyset$ for every such $k$.
Define $\phi(x) = \omega(k) / |T_k|$ if $x$ is in $T_k$ for some $k \in [L]$, and $\phi(x) = 0$ otherwise, where $\omega$ is obtained by applying Lemma \ref{lem:or-dual}. 
Note that $\phi(x)$ is well-defined since $|T_k| \neq 0$ for all $k \in [L]$, and each $x \in \{-1, 1\}^n$ is in $T_k$ for at most one value of $k$.

We check:
\[\sum_{x \in T_1} \phi(x) - \sum_{x \in T_2} \phi(x) = \omega(1) -\omega(2) \ge 1 - \delta,\]
where the inequality holds by Parts 1 and 2 of Lemma \ref{lem:or-dual}. Moreover,
\[\sum_{x \in B} |\phi(x)| = \sum_{k = 3}^{\edit{L}} |\omega(k)| \le \delta,\] where the inequality holds by combining Parts 1-3 of Lemma \ref{lem:or-dual}.
Thus,
\[\sum_{x \in T_1} \phi(x)-\sum_{x \in T_2} \phi(x)-\sum_{x \in B} |\phi(x)| \ge 1 - 2\delta.\]
Second,
\[\sum_{x \in T_1\cup T_2 \cup B} |\phi(x)| = \sum_{k=1}^{\edit{L}} |\omega(k)| = 1,\] where the final equality holds by Part 3 of Lemma \ref{lem:or-dual}.

\edit{Finally, let $d = \zeta \sqrt{\delta L}$ where $\zeta$ is as in the statement of Lemma \ref{lem:or-dual}, and let $S \subseteq [n]$ with $|S| \le d$.}
We must show that $\sum_{x \in T_1 \cup T_2 \cup B} \phi(x) \chi_S(x) =0$. Note that:
\[\sum_{x \in T_1 \cup T_2 \cup B} \phi(x) \chi_S(x) = \sum_{k=1}^{\edit{L}} \sum_{x \in T_k} \phi(x) \cdot \chi_S(x)
=  \sum_{k=1}^{\edit{L}} \sum_{x \in T_k} \left(\omega(k)/\edit{|T_k|}\right) \cdot \chi_S(x) = \sum_{k=1}^{\edit{L}}  \omega(k) \cdot \E_{x \in T_k}[\chi_S(x)],\]
where the first equality holds because $\phi(x)=0$ if $x$ is not in $T_k$ for some $k \in [L]$. 

By Lemma \ref{lem:symmetrization}, there is a trivariate polynomial $P$ of total degree at most $d$ such that $P(m, a, b) = \E_{x \in R_{m, a, b}}[\chi_S(x)]$ 
for all valid triples $(m, a, b)$. In particular, since $k|N$ for all $k \in \edit{[L]}$, $q(k) := P(N, k, 1)$ is a \emph{univariate} polynomial in $k$ such that $q(k) = E_{x \in T_k}[\chi_S(x)]$ for all $k \in \edit{[L]}$. 
Hence, Part 4 of Lemma \ref{lem:or-dual} implies that \[\sum_{k=1}^{\edit{L}}  \omega(k) \cdot \E_{x \in T_k}[\chi_S(x)] = 0. \]

\end{proof}

\section{An $\Omega(N^{1/3})$ Lower Bound for the Collision Function}
\label{sec:mainlb}

The following lemma is a refinement of \cite[Proposition 10]{bt14}, which constructed an explicit dual polynomial
for $\MAJ$.

\begin{lemma} \label{lem:maj-dual}
There exists a constant $\rho > 0$ for which the following holds. Let $\delta \in (0, 1)$, $N > 0$ an even integer, and $k \in [N]$. Then there is an explicit $\eta_k: \edit{[[N]]} \to \R$ such that
\begin{enumerate}
\item $\eta_k$ is supported on $\{2k, 4k, \dots, 2\lfloor N/2k \rfloor k\} \cup \{N/2\}$
\item $\eta_k(N/2) > \frac{1 - \delta}{2}$
\item $\sum_{r = 0}^{N} |\eta_k(r)| = 1$
\item For every polynomial $p: \{0, \dots, N\} \to \R$ of degree $d \le \rho \sqrt{\delta} N / k$, we have \edit{$\sum_{r = 0}^N  p(r) \eta_k(r) = 0$}.
\end{enumerate}
\end{lemma}

\begin{proof}
Throughout the proof, we assume for simplicity that $N/2$ is not a multiple of $2k$. The analysis when $N/2$ \edit{is a} multiple of $2k$ is similar.

Let $c = \lceil \frac{10}{\sqrt{\delta}} \rceil$ and $t = 2\lfloor N/(4k) \rfloor k$ and define the set
\[S = \{t \pm 2c \ell k \colon 0 \le \ell \le \lfloor t/(2ck) \rfloor\}.\]
Note that $|S| = \Omega(N/ck)$. We claim that $\pi_S(i) := \prod_{j \in S, j \ne i} |j - i|$ is minimized at $i = t$. Notice that translating all points in $S$ by a constant does not affect $\pi_S(i)$, and scaling all points in $S$ 
by a constant does not affect $\text{argmin}_i \pi_S(i)$. Thus, it is enough to show that $\pi_{S^*}(i)$ is minimized at $i=0$ for the set $S^* = \{\pm \ell \colon \ell \le t\}$. In this case, $\pi_{S^*}(i)$ takes the simple form $(t - i)!(t + i)!$, 
and we see that for all $i \in S^*$,
\[\frac{\pi_{S^*}(0)}{\pi_{S^*}(i)} = \frac{(t!)^2}{(t-i)!(t+i)!} = \frac{t}{t+|i|} \cdot \frac{t-1}{t+|i|-1} \cdot \dots \cdot \frac{t - |i| + 1}{t + 1}\]
is a product of terms smaller than $1$, so $\pi_{S^*}(i)$ is indeed minimized at $i = 0$.

Now let $T = S \cup \{t - 2k, N/2\}$ and define the function
\[\hat{\eta}(r) = {N \choose r} \frac{(2ck)^{2h}(h!)^2 (2k) (N/2 - t)}{N!} \prod _{j \in [[N]] \setminus T} (r - j) = \frac{(2ck)^{2h}(h!)^2 (2k) (N/2 - t)}{\prod_{j \in T \setminus \{r\}} (r - j)}\]
where $h = \lfloor t/2ck \rfloor$. The normalization is chosen so that $|\hat{\eta}(t)| = 1$. 

The reason that we include \emph{both} $(r - (t-2k))$ and $(r - (N/2))$ in \edit{the denominator of $\hat{\eta}$} is to ensure that the rate of 
decay of $\hat{\eta}(r)$ is at least quadratic as $r$ moves away from $t$. This will ultimately allow us to show that a large fraction of the $\ell_1$ mass of $\hat{\eta}$ comes from the point $r = N/2$.

For $r = t - 2k$, the mass $|\hat{\eta}(r)|$ is
\begin{align*}
\frac{(2ck)^{2h}(h!)^2 (2k) (N/2 - t)}{2k(N/2 - t + 2k) \prod_{\ell = 1}^h (2ck\ell - 2k)(2ck\ell + 2k)} &= \frac{(N/2-t)}{N/2-t+2k} \prod_{\ell = 1}^h \left(1 + \frac{1}{(c\ell)^2 - 1}\right) \\
&\edit{\le \frac{1}{2}\exp\left(\sum_{\ell = 1}^h \frac{2}{c^2\ell^2}\right)} \\
&\edit{\le \frac{1}{2}\exp\left(\frac{\pi^2}{3c^2}\right)} \\
&<\frac{1+\delta}{2},
\end{align*}
\edit{where the first inequality holds because $N/2 - t \le 2k$, combined with the fact that $\prod_{\ell = 1}^h (1 + a_\ell) \le \exp(\sum_{\ell = 1}^h a_\ell)$ for nonnegative $a_\ell$.}

For $r = N/2$, we get
\begin{align*}
\edit{|\hat{\eta}(r)|} = \frac{(2ck)^{2h}(h!)^2 (2k) (N/2 - t)}{(N/2 - t)(N/2 - t + 2k) \prod_{\ell = 1}^h (2ck\ell + (N/2-t))(2ck\ell - (N/2-t))} \\
= \frac{2k}{N/2-t+2k} \prod_{\ell = 1}^h \left(\frac{(2ck\ell)^2}{(2ck\ell)^2 - (N/2 -t)^2}\right) \ge \frac{1}{2}.
\end{align*}

Now we analyze the remaining summands, and show that their total contribution is much smaller than $1$. Recall that the choice $i = t$ minimizes $\pi_S(i)$, and that $\pi_S(t)=(2ck)^{2h}(h!)^2$. Therefore,
\[|\hat{\eta}(t + 2ck\ell)| = \frac{(2ck)^{2h}(h!)^2 (2k)(N/2-t)}{\prod_{j \in T \setminus \{t+2ck\ell\}} |t + 2ck\ell - j|} \le \frac{2k(N/2 - t)}{|2ck\ell + 2k||2ck\ell-(N/2 - t)|} \le \frac{1}{c^2\ell^2 - 1},\]
where the final inequality exploits the fact that $N/2 -t < 2k$. Similarly,

\[|\hat{\eta}(t - 2ck\ell)| = \frac{(2ck)^{2h}(h!)^2 (2k)(N/2-t)}{\prod_{j \in T \setminus \{t+2ck\ell\}} |t - 2ck\ell - j|} \le \frac{2k(N/2 - t)}{|2ck\ell - 2k||2ck\ell+(N/2 - t)|} \le \frac{1}{c^2\ell^2 - 1}.\]

We can use this quadratic decay to bound the total $\ell_1$ mass of the points outside of $\{t-k, t, N/2\}$:
\[\sum_{j \in S \setminus \{t\}} |\hat{\eta}(j)| \le \sum_{\ell\ne 0} \frac{1}{(c^2 \ell^2 -1)}\le \frac{2}{c^2-1} \cdot \frac{\pi^2}{6} < \frac{\delta}{2}.\]

Now let $\eta_k(r) = (-1)^{\edit{r + h + N/2}}\hat{\eta}(r) / \|\hat{\eta}\|_1$. Since $\hat{\eta}$ is supported on $T \subseteq \{2k, 4k, \dots, 2\lfloor N/2k \rfloor k\} \cup \{N/2\}$, the function $\eta_k$ is as well, giving the first claim. Moreover,
\[\eta_k(N/2) \ge \frac{1/2}{(1/2 + \delta/2) + 1/2 + \delta/2} \ge \frac{1 - \delta}{2}.\]
This yields the second claim about $\eta_k$. The third claim follows immediately from the definition. Finally, let $p$ be a polynomial of degree strictly less than $|T|$ (where $|T| \ge \rho N/k$ for a constant $\rho$). Then

\begin{align*}
\sum_{r=0}^{N} p(r) \eta_k(r) &= \sum_{r = 0}^N p(r) \frac{(-1)^\edit{r + h + N/2}}{\|\hat{\eta}\|_1} {N \choose r} \frac{(2ck)^{2h}(h!)^2 (2k) (N/2 - t)}{N!} \prod _{j \in \edit{[[N]]} \setminus T} (r - j) \\
&=\sum_{r=0}^{N} (-1)^r {N \choose r} q(r)
\end{align*}
for a polynomial $q$ of degree strictly less than $N$. This is equal to zero by Lemma \ref{lem:combinatorial}, giving the final claim.
\end{proof}

We obtain our dual polynomial $\psi$ for the $\col_{N, R}$ as a linear combination of two simpler functions $\psi_1$ and $\psi_2$. These functions have the following properties.

\begin{lemma} \label{lem:intermediate}
Let $N > 0$ be an integer multiple of 4. For $R \ge N$, there exist explicit $\psi_1, \psi_2 \colon \{-1, 1\}^n \to \R$ and $d = \Omega(\delta^{1/3}N^{1/3})$ such that
\begin{enumerate}
\item $\sum_{x \in T_1}\psi_1(x) > \frac{1-\delta}{2}$.
\item $-\sum_{x \in T_2} \psi_2(x) > \frac{1-\delta}{2}$.
\item $\edit{\|\psi_1\|_1 = \|\psi_2\|_1 = 1}$.
\item $\sum_{x\in T_2} |\psi_1(x)| = \sum_{x \in T_1} |\psi_2(x)| = 0$.
\item $\psi_1, \psi_2$ have pure high degree at least $d$.
\item $\sum_{x \in T_1}\psi_1(x) = \sum_{x \in R_{N/2, 2, 1}} \psi_2(x)$. 
\item $\sum_{x \in R_{N/2, 2, 1}} \psi_1(x) = \sum_{x \in T_2} \psi_2(x)$.
\item $\psi_1$ and $\psi_2$ are each constant on each set $R_{m, a, b}$ when $(m, a, b)$ is valid.
\end{enumerate}
\end{lemma}

Together, they yield the desired dual polynomial for $\col_{N, R}$.

\begin{theorem} \label{thm:collision_main}
Let $N > 0$ be an integer multiple of 4. For $R \ge N$, there exists an explicit $(1-6\delta, d)$-dual polynomial $\psi$ for $\col_{N, R}$ for $d = \Omega(\delta^{1/3}N^{1/3})$.
\end{theorem}

\begin{remark}
The dependence of the lower bound Theorem \ref{thm:collision_main} on both parameters $\delta$ and $N$ for $\frac{1}{N} \le \delta \le \frac{1}{10}$, is tight up to a logarithmic factor in the size of the range. We show this in Appendix \ref{app:tightness} by constructing an explicit approximating polynomial for $\col_{N, R}$ of the appropriate degree, by building on the ideas underlying the quantum query algorithm of Brassard et al. \cite{brassard}.
\end{remark}

\begin{proof}[Proof of Theorem \ref{thm:collision_main}, assuming Lemma \ref{lem:intermediate}]
Let $a = \sum_{x \in T_1}\psi_1(x)$ and let $b = \sum_{x \in T_2} |\psi_2(x)| = -  \sum_{x \in T_2} \psi_2(x)$, where $\psi_1$ and $\psi_2$ are as in Lemma \ref{lem:intermediate}. Let $\psi(x) = a\psi_1(x) + b\psi_2(x)$. By Property 5 of Lemma \ref{lem:intermediate} and Lemma \ref{lem:phd-sum} below, $\psi$ also has pure high degree at least $d$. So we need
only show that $\psi$ has correlation at least $1-6\delta$ with $\col_{N, R}$. 
To this end, note that
\begin{enumerate}
\item $\sum_{x \in T_1} \psi(x) = a^2  > \frac{(1-\delta)^2}{4}$. This inequality uses Properties 1 and 3 of Lemma \ref{lem:intermediate}.
\item $-\sum_{x \in T_2} \psi(x)  = b^2 > \frac{(1-\delta)^2}{4}$. This inequality uses Properties 2 and 3 of Lemma \ref{lem:intermediate}.
\item $\sum_{x \in B} |\psi(x)| \le a \sum_{x \in B \setminus R_{N/2, 2, 1}} |\psi_1(x)| + b \sum_{x \in B \setminus R_{N/2, 2, 1}} |\psi_2(x)| \le (a+b)\delta$.
Here, the first inequality exploits the fact that 
\begin{align} \label{eq:pleasethishastoenddiediedie} \sum_{x \in R_{N/2, 2, 1}} |\psi(x)| = \sum_{x \in R_{N/2, 2, 1}} |a \cdot \psi_1(x) + b \cdot \psi_2(x)| = 0.\end{align}
The last equality in \eqref{eq:pleasethishastoenddiediedie} holds because, for all $x \in R_{N/2, 2, 1}$, 
{\small
\begin{flalign*} & a\cdot \psi_1(x) + b \cdot \psi_2(x)  \\
&= \left(\sum_{x' \in T_1} \psi_1(x') \right) \cdot \psi_1(x)  +  \left(-\sum_{x' \in T_2} \psi_2(x') \right) \psi_2(x)\\
& =  \left(\sum_{x' \in R_{N/2, 2, 1}} \psi_2(x') \right) \cdot \psi_1(x)  +  \left(-\sum_{x' \in R_{N/2, 2, 1}} \psi_1(x') \right) \psi_2(x) \\
&  =  \left(\sum_{x' \in R_{N/2, 2, 1}} \psi_2(x') \right) \left(\frac{1}{|R_{N/2, 2, 1}|} \sum_{x' \in R_{N/2, 2, 1}} \psi_1(x')\right)\!+\!  \left(-\sum_{x' \in R_{N/2, 2, 1}} 
\psi_1(x') \right)\left(\frac{1}{|R_{N/2, 2, 1}|} \sum_{x' \in R_{N/2, 2, 1}} \psi_2(x')\right)\\
& =0,\end{flalign*}}

where the second equality exploited Properties 6 and 7 of Lemma \ref{lem:intermediate}, \edit{and the last equality exploited Property 8}.
\end{enumerate}
Thus, the correlation of $\psi$ with $\col_{N, R}$ is
\begin{align*}
\sum_{x \in T_1} \psi(x) - \sum_{x \in T_2} \psi(x) - \sum_{x \in B} |\psi(x)| &\ge a^2 + b^2 - (a+b)\delta \\
&\ge \frac{1}{2}-2\delta \geq (1-6\delta)\cdot \|\psi\|_1,
\end{align*}
where the final inequality holds because $\|\psi\|_1 \le a^2 + b^2 + (a+b)\delta \le \frac{1}{2} + \delta$. 
\end{proof}

\begin{lemma} \label{lem:phd-sum}
Let $\psi_1, \psi_2 : \bits^n \to \bits$ each have pure high degree at least $d$. Then $\psi = \psi_1 + \psi_2$ also has pure high degree at least $d$.
\end{lemma}

\begin{proof}
Let $S \subseteq [n]$ with $|S| \le d$. Then
\[\sum_{x \in \bits^n} \psi(x) \chi_S(x) = \sum_{x \in \bits^n} \psi_1(x) \chi_S(x) + \sum_{x \in \bits^n} \psi_2(x) \chi_S(x) = 0.\]
\end{proof}

\begin{proof}[Proof of Lemma \ref{lem:intermediate}]

Let $\zeta$ be the constant from Lemma \ref{lem:or-dual}, let $\rho$ be the constant from Lemma \ref{lem:maj-dual}, and let $\delta' = 1/2$. Set $K = 2(\rho N / \zeta)^{2/3}(\delta'/\delta)^{1/3}$. Let $d = \frac{1}{2}\rho^{1/3}\zeta^{2/3}(\delta')^{1/6}\delta^{1/3}N^{1/3} = \Omega(\delta^{1/3}N^{1/3})$, noting that $d \le \zeta (\delta/8)^{1/2} K^{1/2}$ and $d \le \rho (\delta')^{1/2}N / k$ for every $k \le K$. Let $\omega \colon \{1, \dots, K\} \to \R$, with correlation constant $\delta/8$, and $\eta_3, \dots, \eta_K \colon \{1, \dots, N\} \to \R$, with correlation constant $\delta'$, be as in the conclusions of those lemmas.

We start by defining a function $\Psi(m, k)$ as follows.
\[\Psi(m, k) = \omega(k) \cdot \mathbbm{1}_{m = N/2}- \frac{\omega(k)}{\eta_k(N/2)} \mathbbm{1}_{k \ge 3} \cdot \eta_k(m).\]

Here,
\[\mathbbm{1}_{m = N/2} = \begin{cases}
1 \quad \text{if } m = N/2, \\
0 \quad \text{otherwise,}
\end{cases} \quad\quad and \quad\quad
\mathbbm{1}_{k \ge 3} = \begin{cases}
1 \quad \text{if } k \ge 3, \\
0 \quad \text{otherwise.}
\end{cases}\]

We first show how to use $\Psi$ to construct the polynomial $\psi_1$. Analogously to our construction of $\phi$, we want $\psi_1$ to place a total weight of $\Psi(m, a)$ on each set $R_{m, a, 1}$. Recall from our overview in Section \ref{sec:overview} that we think of $\psi_1 = \psi' + \sum_{\text{invalid triples } (N/2, a, 1)} \psi''_{N/2, a, 1}$, where $\psi'$ looks like the simpler ``first stage'' dual polynomial $\phi$ from our informal overview (which we constructed in Section \ref{sec:firstlb}) and each $\psi''_{N/2, a, 1}$ cancels out the weight $\phi$ places on values of $k$ the do not divide $N$. 
This structure underlies our construction of $\Psi$, where we add multiples of the polynomials $\eta_k(m)$ to cancel out the weight $\omega(k)$ places on invalid triples.

Now we construct and analyze the polynomial $\psi_1$. Define 
\[\hat{\psi}_1(x) = \begin{cases} \Psi(m, k) / |R_{m, k, 1}| & \text{ if } x \in R_{m, k, 1} \setminus T_1\\ 
\Psi(N/2, 1) / |T_1| & \text{ if } x \in T_1\\
0& \text{ otherwise.}
\end{cases}\] 
Notice that $\hat{\psi}_1$ is well-defined, because any $x \not\in T_1$ is in $R_{m, k, 1}$ for at most one triple $(m, k, 1)$.
%
We collect several calculations with $\hat{\psi}_1$. First,
\[\sum_{x \in T_{1}} \hat{\psi}_1(x) = \Psi(N/2, 1) = \omega(1) > \frac{1 - \delta/8}{2},\]
\[-\sum_{x \in R_{N/2, 2, 1}} \hat{\psi}_1(x) = -\Psi(N/2, 2) = -\omega(2) > \frac{1 - \delta/8}{2},\]
and
\begin{align*}
\sum_{x \in B \setminus R_{N/2, 2, 1}} |\hat{\psi}_1(x)| &= \sum_{\{(m, k) \colon k | m\} \setminus \{(N/2, 1), (N/2, 2)\}} |\Psi(m, k)| \\
&= \sum_{k = 3}^K \left| \frac{\omega(k)}{\eta_k(N/2)} \right |\sum_{i = 1}^{\lfloor N/2k \rfloor} |\eta_k(2ki)|\\
&\le 4\sum_{k = 3}^K |\omega(k)| \\
&\le \frac{\delta}{2},
\end{align*}
where the penultimate inequality exploits Properties 2 and 3 of Lemma \ref{lem:maj-dual}, and the final inequality exploits Properties 1-3 of Lemma \ref{lem:or-dual}.

Noting that $|\omega(1)| + |\omega(2)| \leq 1$, it follows that $\|\hat{\psi}_1\|_1 \le 1 + \delta/2$. So setting $\psi_1 = \hat{\psi}_1/ \|\hat{\psi}_1\|_1$, it is immediate that $\psi_1$ satisfies the first three 
properties in the statement of the lemma. $\psi_1$ also satisfies the fourth property, since for any $x \in T_2$, $\psi_1(x) = \Psi(N, 2)/|T_2| = 0$.

Now we will show that $\hat{\psi}_1$, and hence $\psi_1$, has pure high degree at least $d$. We require two observations.
\begin{itemize}
\item $\Psi$ is supported on $(m, k)$ for which $k | m$. To see this, note first that
for any $k \geq 3$, $\Psi(N/2, k) = \omega(k) - \frac{\omega(k)}{\eta_k(N/2)} \cdot \eta_k(N/2) = 0$.
The claim now follows from Property 1 of Lemma \ref{lem:maj-dual}, combined with the fact that $2 | N$. 
 
\item $\Psi(m, 1)$ is nonzero only for $m = N/2$, and hence $\sum_{x \in T_1} \hat{\psi}_1(x) = \Psi(N/2, 1) = \sum_{m=1}^N \Psi(m, 1)$.
\end{itemize}

Fix any $S \subseteq [n]$ with $|S| \edit{\le} d$. Let $Q(m, k)$ be a polynomial of degree at most $d-1$
in each variable such that, for all pairs $(m, k)$ with $k | m$, $Q(m, k)=\E_{x \in R_{m, k, 1}}[\chi_S(x)]$. The existence of such a bivariate polynomial $Q$ is guaranteed by Lemma \ref{lem:symmetrization}.
Then the previous two observations together imply that:
\begin{align} \label{eq:nomoreseriously}
\sum_{x \in \bits^n}  \hat{\psi}_1(x) \chi_S(x) &= \sum_{m = 1}^N \sum_{k=1}^N \Psi(m, k) Q(m, k).
\end{align}
We remark that a key point is the derivation of \eqref{eq:nomoreseriously} is that we have no control over the evaluations $Q(m, k)$ when $k$ does
not divide $m$, yet this is rendered irrelevant because $\Psi(m, k) = 0$ for all such pairs. 

The right hand side of \eqref{eq:nomoreseriously} equals:
\begin{align}  \label{eq:soveryveryclose} 
\sum_{k=1}^K \omega(k) Q(N/2, k) - \sum_{k = 3}^K \frac{\omega(k)}{\eta_k(N/2)}\left(\eta_k(N/2) Q(N/2, k) + \sum_{i = 1}^{\lfloor N/2k \rfloor} \eta_k(2ik) Q(2ik, k)\right).
\end{align}

The first sum in \eqref{eq:soveryveryclose} is zero by Lemma \ref{lem:or-dual} since $Q(N/2, k)$ is a polynomial of degree at most \edit{$d$} in $k$. The second sum is also zero because for each fixed $k$, $Q(r, k)$ is a polynomial of degree at most \edit{$d$} in the variable $r$, and hence the term in parentheses is zero by Lemma \ref{lem:maj-dual} (Parts 1 and 4). Thus $\hat{\psi}_1$ has pure high degree at least $d$.

The construction of $\psi_2$ is similar. This time, we let
\[\hat{\psi}_2(x) = \begin{cases}
\Psi(m, k) / |R_{m, k, 2}| & \text{ if } x \in R_{m, k, 2} \setminus T_2 \\
\Psi(N/2, 2) / |T_2| & \text{ if } x \in T_2 \\
0 & \text{ otherwise.}
\end{cases}\]
Note that $\hat{\psi}_2$ is well-defined, because any $x \not\in T_2$ is in $R_{m, k, 2}$ for at most one triple $(m, k, 2)$. We define $\psi_2 = \hat{\psi}_2/ \|\hat{\psi}_2\|_1$.
Showing that $\psi_2$ satisfies Properties 1-4 of the lemma follows from the same calculations we used for $\psi_1$. 

To show that
$\psi_2$ has pure high degree at least $d$, we require the following additional observations.
\begin{itemize}
\item $\Psi$ is supported on pairs $(m, k)$ for which $k | m$ and $2 | (N-m)$. To see the latter property, note that if $\Psi(m, k) \neq 0$, then 
$m$ is even (\edit{this holds because $N/2$ is even, which follows from our requirement that $N$ is a multiple of 4}), and hence $N-m$ is as well.
\item  $\Psi(m, 2)$ is nonzero only for $m = N/2$. It follows that  $\sum_{x \in T_2} \hat{\psi}_2(x) = \Psi(N/2, 2) = \sum_{m=1}^N \Psi(m, 2)$.
\end{itemize}
With these observations in hand, showing that $\psi_2$ has pure high degree \edit{$d$} then follows from calculations analogous to the ones we used for $\psi_1$.

Finally, the fact that $\psi_1$ and $\psi_2$ satisfy Properties 6, 7\edit{, and 8} of the lemma follows from their definitions, combined with the fact that $R_{N/2, 1, 2} = R_{N/2, 2, 1}$.
In fact, $\sum_{x \in T_1} \psi_1(x)$ equals $\Psi(N/2, 1)$, while $\sum_{x \in R_{N/2, 2, 1}} \psi_2(x)$ also equals $\Psi(N/2, 1)$, giving Property 6. Similarly, 
$\sum_{x \in T_2} \psi_2(x)$ equals $\Psi(N/2, 2)$, while $\sum_{x \in R_{N/2, 1, 2}} \psi_1(x)$ also equals $\Psi(N/2, 1)$. This completes the proof.

\end{proof}

\section{A Dual Polynomial for Element Distinctness}
\label{sec:ed}

We first recall the reduction from Collision to Element Distinctness given in \cite{aaronsonshi}.\footnote{While the reduction given in Aaronson and Shi's paper is stated in terms of quantum query algorithms, it is straightforward to rephrase the reduction in terms of approximating polynomials instead.} The reduction shows how to turn a polynomial $p$ approximating $\ed_{M, R}$ into a polynomial $q$ approximating $\col_{N, R}$, with $N \approx M^2$ and $\deg q \le \deg p$.

We illustrate the reduction for $N = M^2/12$. Let $p : \{-1, 1\}^m \to \{-1, 1\}$ be an $(1/6)$-approximation of $\ed_{M, R}$, with $m = M\log R$. Define a polynomial $q:\bits^n \to \bits$ for $n = N\log R$ by
\[q(y_1, \dots, y_N) = \frac{1}{{N \choose M}} \sum_{1 \le i_1 < i_2 < \dots < i_M \le N} p(y_{i_1}, y_{i_2}, \dots, y_{i_M}).\]
That is, $q(y)$ is the expected value of $p(x)$ where $x$ is the concatenation of a random subset of $M$ of the blocks $y_1, \dots, y_{N}$. To simplify notation, for a set $S = \{i_1, i_2, \dots, i_M\}$, let $y|_S = (y_{i_1}, y_{i_2}, \dots, y_{i_M})$. Note that $\deg q \le \deg p$. Moreover, since $q$ is an average of values in $[-7/6, 7/6]$, it is always in $[-7/6, 7/6]$ itself. To finish arguing that $q$ is a $(1/3)$-approximation to $\col_{N, R}$, we consider two cases:

\begin{enumerate}
\item If $y \in T_1$, i.e., $y$ is a 1-to-1 input, then $y|_S$ is always 1-to-1. Hence $p(y|_S) \in [5/6, 7/6]$ for every subset of indices, so $q(y) \in [2/3, 4/3]$.
\item If $y \in T_2$, i.e., $y$ is a 2-to-1 input, then with high probability $y|_S$ is not 1-to-1. This follows from the ``birthday bound'':
\[\Pr_{|S| = M} [\ed(y|_S) = 1] \le \exp(-M^2/4N) \le \frac{1}{12}.\]
Therefore, $q(y) \le (11/12)(-5/6) + (1/12)(7/6) \le -2/3$.
\end{enumerate}

The construction we give in this section takes a dual view of the reduction above. Namely, we show how to transform a dual polynomial $\psi$ for $\col_{N, R}$ into a dual polynomial $\varphi$ for $\ed_{M, R}$, with $M^2 \approx N$. In the primal reduction, we constructed $q(y)$ from $p(x)$ by averaging $p$ over all subsets of size $M$. The right analogue in the dual reduction is to construct $\varphi(x)$ by averaging $\psi(y)$ over a carefully constructed set of \emph{extensions} from $x$ to a longer input $y$. In particular, $\varphi(x)$ averages $\psi(y)$ over all $y$ for which $x$ could have been produced by taking a subset of $M$ blocks of $y$.

We give this reduction formally below.

\begin{theorem}\label{thm:ed-reduction}
Let $\psi : \bits^n \to \bits$ be a $(1-\delta, d)$-dual polynomial for $\col_{N, R}$. Then $\psi$ can be used to construct $\varphi : \bits^m \to \bits$ that is an $(1-2\delta, d)$-dual polynomial for $\ed_{M, R}$ when $M \ge 2\sqrt{N \log (2/\delta)}$.
\end{theorem}

\begin{corollary} \label{cor:ed_lower}
For any $\delta > 0$, there is an explicit $(1-\delta, d)$-dual polynomial for $\ed_{M, R}$ with $d = \Omega((\delta/\log(1/\delta))^{1/3}M^{2/3})$.
\end{corollary}

\begin{remark}
The dependence of Corollary \ref{cor:ed_lower} on $\delta$ is essentially tight for $\delta = O(M^{-2})$. See Appendix \ref{app:tightness} for details.
\end{remark}

\begin{proof}[Proof of Theorem \ref{thm:ed-reduction}]

Given a set $S = \{i_1, \dots, i_M\} \subset [N]$ with $i_1 < i_2 < \dots < i_M$ and a bit string $y = (y_1, \dots, y_N) \in \bits^n$, define the restriction of $y$ to the set $S$, denoted by $y|_S \in \bits^m$, to be the string of length $m = M \log R$ obtained by concatenating the blocks $y_i$ for $i \in S$, i.e., $y|_S = (y_{i_1}, y_{i_2}, \dots, y_{i_M})$. Given a bit string $x \in \bits^m$, define the multiset of extensions of $x$, denoted by $\ext(x)$, to be the ${N \choose M} R^{N - M}$ strings $y \in \bits^n$ where $y|_S = x$ for some $|S| = M$. Restrictions and extensions are related by the equivalence of the multisets:
\[\{(x, y) : x \in \bits^{m}, y \in \ext(x)\} = \{(x, y) : y \in \bits^n, x = y|_S \text{ for some } |S| = m\}.\]

For $x \in \bits^m$, define the polynomial
\[\varphi(x) = \frac{1}{{N \choose M}}\sum_{y \in \ext(x)} \psi(y).\]
Let $\varphi(x) = 0$ for $x \notin \bits^m$. We claim that $\varphi$ is a good dual polynomial for the Element Distinctness function $\ed$, which requires us to show
\begin{enumerate}
\item $\sum_{x \in \{-1, 1\}^m} \varphi(x) \ed(x) > (1-2\delta)\cdot \sum_{x \in \{-1, 1\}^m} |\varphi(x)|$
\item $\sum_{x \in \{-1, 1\}^m} \varphi(x)\chi_S(x) = 0$ for all $|S| \le d$
\end{enumerate}

To verify the first property, define
\[A(y) = \frac{1}{{N \choose M}} \sum_{|S| = M} \ed(y |_S).\]
We collect a few observations about $A$.

\begin{enumerate}
\item $|A(y)| \le 1$ for all $y$.
\item If $y \in T_1$, then $A(y) = 1$.
\item If $y \in T_2$, then
\[\Pr_{|S| = M} [\ed(y|_S) = 1] \le \exp(-M^2/4N).\]
\eat{
\begin{align*}
\Pr_{|S| = M} [\ed(y|_S) = -1] &= \frac{{N/2 \choose M}2^M}{{N \choose M}} \\
&= \frac{(N/2)!(N-M)!2^M}{(N/2-M)!N!} \\
&= \frac{N \cdot (N-2) \cdot \dots \cdot (N - 2M + 2)}{N \cdot (N-1) \cdot \dots \cdot (N - M + 1)} \\
&\le \left(\frac{N - M}{N - M/2}\right)^{M/2} \\
&\le \exp\left(\frac{-M^2}{4N - 2M}\right) \\
&\le \exp(-M^2/4N).
\end{align*}
}
Hence,
\[A(y) \le -1 + 2\exp(-M^2/4N) \le -1 + \delta.\]
\end{enumerate}
Therefore we get
\begin{align*}
\sum_{x \in \bits^m} \varphi(x) \ed(x) &= \frac{1}{{N \choose M}} \sum_{x \in \bits^m} \sum_{y \in \ext(x)} \psi(y)\ed(x) \\
&= \frac{1}{{N \choose M}} \sum_{y \in \bits^n} \sum_{|S| = M} \psi(y)\ed(y|_S) \\
&= \sum_{y \in \bits^n} A(y) \psi(y) \\
&\ge \left(\sum_{y \in T_1} \psi(y) - \sum_{y \in T_2} \psi(y) - \sum_{y \in B} |\psi(y)|\right) - \delta \sum_{y \in T_2} |\psi(y)|\\
&\ge (1-2\delta) \sum_{y \in \bits^n} |\psi(y)| \\
&= (1-2\delta) \sum_{y \in \bits^n} \frac{1}{{N \choose M}} \sum_{|S| = M} |\psi(y)| \\
&\ge \frac{1-2\delta}{{N \choose M}}\sum_{x \in \bits^m} \sum_{y \in \ext(x)} |\psi(y)| \\
&\ge (1-2\delta) \sum_{x \in \bits^m} |\varphi(x)|.
\end{align*}

For the second property, let $T$ be a subset of $[N]$ with $|T| \le d$. Then
\begin{align*}
\sum_{x \in \bits^m} \varphi(x) \chi_T(x) &= \frac{1}{{N \choose M}} \sum_{x \in \bits^m} \sum_{y \in \ext(x)} \psi(y)\chi_T(x) \\
&= \frac{1}{{N \choose M}} \sum_{|S| = M} \sum_{y \in \bits^n} \psi(y)\chi_T(y|_S) \\
&= \frac{1}{{N \choose M}} \sum_{|S| = M} \sum_{y \in \bits^n} \psi(y)\chi_{T|_S}(y) \\
&= 0,
\end{align*}
where $T|_S$ denotes the subset of $T$ contained in the blocks specified by $S$.

\end{proof}

\section{On Complementary Slackness}
\label{sec:complementaryslackness}
Recalling that any bounded-error quantum query algorithm can be converted into an approximating polynomial \cite{beals}, the collision-finding algorithm of Brassard, H{\o}yer, and Tapp \cite{brassard} yields an explicit, asymptotically optimal approximating polynomial for $\col_{N, R}$. We describe this polynomial $p$ below. 

Recall that any approximating polynomial for $\col_{N, R}$ represents a feasible solution to the primal linear program considered in Section \ref{sec:prelims}. If the polynomial $p$ were an \emph{exactly} optimal $\eps$-approximation
for $\col_{N, R}$, then complementary slackness would imply that the optimal dual polynomial $\psi$ for $\col_{N, R}$ is supported on the points corresponding to constraints made tight by $p$. That is, $\psi:\bits^n \to \bits$ is supported on $x \in \bits^n$ for which $|p(x) - \col(x)| = \eps$. 
We refer to these as the \emph{maximum-error points} of $p$. 

While we do not know whether $p$ is an exactly optimal approximating polynomial for $\col_{N, R}$, we might still expect that an approximate version of complementary slackness might holds,
in the sense that a ``good'' dual polynomial should place all or most of its weight on points that are ``nearly'' maximum-error points of $p$. Indeed, this intuition has proven accurate for all of the dual polynomials
constructed in prior work, including for symmetric functions (see \cite[Section 4.5]{bt13}), block-composed functions (see \cite[Section 1.2.4]{thaler}), and the intersection of two majorities \cite{sherstovfocs}. 
Below, we argue that $k$-to-1 inputs are nearly maximum-error points for $p$, which explains why our dual polynomials for collision are supported on inputs that are roughly $k$-to-1, in addition to why these inputs play a prominent role in the original symmetrization-based proof. 

\paragraph{An asymptotically optimal approximation $p$ for $\col_{N, R}$.}
For a subset $S \subset [N]$, define $\cross_S \colon \{-1, 1\}^n \rightarrow \mathbb{R}$ via:
\[\cross_S(x_1, \dots, x_N) = |\{i \in S, j \notin S: x_i = x_j\}| = \sum_{i \in S, j \not\in S} \EQ(x_i, x_j), \]
where $\EQ$ denotes the equality function.
That is, $\cross_S(x)$ counts the number of cross-collisions between indices in $S$ and indices outside of $S$. 
Notice that $\EQ(x_i, x_j)$ is a function of only $2 \cdot \log R$ variables, and hence $\cross_S(x_1, \dots, x_N)$ is exactly computed by a polynomial of degree $2 \cdot \log R$.

In addition, for a subset $S \subset [N]$, define the function $\mathbb{I}_{\ED, S}(x_1, \dots, x_N)$ to be 1 if $x_i \neq x_j$ for all pairs $i, j \in S$ with $i \neq j$, and 0 otherwise.
That is, $\mathbb{I}_{\ED, S}$ indicates whether $x$ is 1-to-1 on the indices in $S$. Notice that $\mathbb{I}_{\ED, S}$ is a function of only $|S| \cdot \log R$ variables,
and hence is exactly computed by a polynomial of degree $|S| \cdot \log R$.

For the remainder of the discussion, let $r=N^{1/3}$ -- we focus on the quantity $\cross_S(x)$ when $|S| = r$. We will need the following simple observations.

\begin{enumerate}
\item If $x \in T_1$, i.e., $x$ is a 1-to-1 input, then $\cross_S(x) = 0$ and $\mathbb{I}_{\ED, S}(x) = 1$ for any $S$. 
\item If $x \in T_2$, i.e., $x$ is a 2-to-1 input, then $\mathbb{I}_{\ED, S}(x) = 1 \Longrightarrow \cross_S(x) = r$. 
\item If $x \in T_2$, then, over the random choice of $S$, $\mathbb{I}_{\ED, S}(x) = 0$ with probability at most $(N/2) \cdot (r/N)^2 \leq N^{-1/3}$.
\item For all $x \in \{-1, 1\}^n$, $\mathbb{I}_{\ED, S}(x) =1 \Longrightarrow \cross_S(x) \leq N-r$.
\end{enumerate}


Let $T_d: \mathbb{R} \rightarrow \R$ denote the degree-$d$ Chebyshev polynomial of the first kind. This polynomial has the following properties:
\begin{itemize}
\item $T_d(x) \in [-1, 1]$ for $x \in [-1, 1]$.
\item $T_d(1 + M/d^2) \ge 10$ for a constant $M$ independent of $d$.
\item The extreme points of $T_d$ in $[-1, 1]$ are the degree-$d$ \emph{Chebyshev nodes}, which take the form $\cos(i \pi/d)$ for $0 \leq i \leq d$.
\end{itemize}

Truncating the Taylor expansion of $\cos(x)=1-x^2/2+\dots$ after the quadratic term, one sees that the Chebyshev nodes are well-approximated via the expression $\cos(i \pi/d) \approx 1-(ci^2/d^2)$ for some constant $c$. 

Applying an appropriate affine transformation to $T_d$, we obtain a polynomial $A_d$ with the following properties:
\begin{itemize}
\item $A_d(0) = 1$.
\item $A_d(i) \in [-1, -3/4]$ for all real numbers $i \in [1, d^2/M]$.
\item $A_d(i) \in [-1, 1]$ for all real numbers $i \in [0, d^2/M]$. 
\item The extreme points of $A_d$ are well approximated by the points $c \cdot i^2$ for $i \in \{0,1, \dots, \lfloor d \cdot M^{-1/2}\rfloor\}$.
\end{itemize}

Let $p_S(x) = \mathbb{I}_{\ED, S}(x) \cdot A_d(\cross_S(x_1, \dots, x_N)/r)$ for $d = 100 \cdot M \cdot N^{1/3}$, and let 
$$p(x) = \E_{|S| = r}[p_S(x)] = \frac{1}{{N \choose r}} \sum_{|S| = r} p_S(x).$$
Then $p$ is a polynomial of degree $|S| \log R + 2 \cdot d  \cdot \log R = O(N^{1/3} \log R)$. We argue that $p$ approximates $\col_{N, R}$ to error $\eps$ for some $\eps \leq 1/3$.
The analysis falls into three cases.

\begin{itemize}
\item[Case 1:] For $x \in T_1$, $p_S(x) = A_d(0) = -1$ for all $S$, where the first equality follows from Property 1 above. So $p(x)=\E_{|S| = r}[p_S(x)] =1$. 
\item[Case 2:] For $x \in T_2$, $\mathbb{I}_{\ED, S}(x) = 1 \Longrightarrow p_S(x) = A_d(1) \in [-1, -3/4]$, where the equality follows from Property 2 above. 
Meanwhile,  $\mathbb{I}_{\ED, S}(x) \neq 1  \Longrightarrow p_S(x) = 0$. 
Combining these two facts with Property 3 above establishes that  $p(x)=\E_{|S| = r}[p_S(x)] \in [-1, -2/3]$.
\item[Case 3:] For $x \in \{-1, 1\}^n$, $p_S(x) \in [-1, 1]$. This follows from Property 4 above. 
\end{itemize}

\paragraph{Identifying maximum-error points of $p$.} 
For any fixed $S$, the maximum error points of $p_S$ are well-approximated by the $x \in \{1, 1\}^n$ for which the following two equations hold:
\begin{equation} \label{damnitendthis} \cross_S(x) = c\cdot i^2 \cdot r \text{ for some } i \in \{0,1, \dots, \lfloor d \cdot M^{-1/2}\rfloor\} \end{equation}
and
\begin{equation} \label{seriouslyman}
\mathbb{I}_{\ED, S}(x)=1.\end{equation}
(This follows from the fact that the extreme points of $A_d$ are roughly of the form $c \cdot i^2$ for $0 \le i \le d \cdot M^{-1/2}$).

However, the maximum-error points for the averaged polynomial $p(x) = \E_{|S|=r}[p_S(x)]$ are the points $x$ that satisfy \eqref{damnitendthis} and \eqref{seriouslyman} \emph{with high probability} over the choice of $S$. 
Indeed, for these points $x$, the error of $p(x)$ is at least $\eps \cdot (1-o(1)) \approx \eps$.

Consider any $k$ of the form $k=c\cdot i^2 + 1$ for some $i \in \{0, 1, \dots, \lfloor d \cdot M^{-1/2}\rfloor\}$, such that $k = o(N^{1/3})$. Consider any $x \in T_k$; we claim that $x$ satisfies \eqref{damnitendthis} and \eqref{seriouslyman} with probability $1-o(1)$
over choice of $S$. To see this, observe that the probability that
$\mathbb{I}_{\ED, S}(x_S)=0$ is at most $(N/k) \cdot k^2 \cdot \left(r/N\right)^2 = \frac{k \cdot r^2}{N} = o(1)$. 
 And if $\mathbb{I}_{\ED, S}(x_S)\neq 0$, then the number of cross-collisions is exactly
\[\cross_S(x_1, \dots, x_N) = r \cdot (k-1).\]
When $k$ takes the form $k=c\cdot i^2 + 1$, this means that $x$ satisfies \eqref{damnitendthis}. Hence, $x$ has nearly maximal error even for the averaged polynomial $p$.

\paragraph{Acknowledgements.} We are grateful to Guy Kindler, Yaoyun Shi, and Mario Szegedy for several illuminating discussions during the early stages of this work. We also thank Scott Aaronson and Emanuele Viola for helpful comments on an earlier version of this manuscript.

\bibliographystyle{plain}
\bibliography{collision}

\begin{thebibliography}{10}

\bibitem{aaronson}
Scott Aaronson.
\newblock Quantum lower bound for the collision problem.
\newblock In John~H. Reif, editor, {\em Proceedings on 34th Annual {ACM}
  Symposium on Theory of Computing, May 19-21, 2002, Montr{\'{e}}al,
  Qu{\'{e}}bec, Canada}, pages 635--642. {ACM}, 2002.

\bibitem{aaronsontutorial}
Scott Aaronson.
\newblock The polynomial method in quantum and classical computing.
\newblock In {\em 49th Annual {IEEE} Symposium on Foundations of Computer
  Science, {FOCS} 2008, October 25-28, 2008, Philadelphia, PA, {USA}}, page~3.
  {IEEE} Computer Society, 2008.

\bibitem{retrospective}
Scott Aaronson.
\newblock The collision lower bound after 12 years.
\newblock In {\em Qstart Conference}, 2013.

\bibitem{aaronsonshi}
Scott Aaronson and Yaoyun Shi.
\newblock Quantum lower bounds for the collision and the element distinctness
  problems.
\newblock {\em J. ACM}, 51(4):595--605, 2004.

\bibitem{adversary2}
Andris Ambainis.
\newblock Quantum lower bounds by quantum arguments.
\newblock {\em J. Comput. Syst. Sci.}, 64(4):750--767, 2002.

\bibitem{ambainis}
Andris Ambainis.
\newblock Polynomial degree and lower bounds in quantum complexity: Collision
  and element distinctness with small range.
\newblock {\em Theory of Computing}, 1(1):37--46, 2005.

\bibitem{qqcpolyseparation}
Andris Ambainis.
\newblock Polynomial degree vs. quantum query complexity.
\newblock {\em Journal of Computer and System Sciences}, 72(2):220 -- 238,
  2006.
\newblock \{JCSS\} \{FOCS\} 2003 Special Issue.

\bibitem{ambainised}
Andris Ambainis.
\newblock Quantum walk algorithm for element distinctness.
\newblock {\em {SIAM} J. Comput.}, 37(1):210--239, 2007.

\bibitem{adversary5}
Howard Barnum, Michael~E. Saks, and Mario Szegedy.
\newblock Quantum query complexity and semi-definite programming.
\newblock In {\em 18th Annual {IEEE} Conference on Computational Complexity
  (Complexity 2003), 7-10 July 2003, Aarhus, Denmark}, pages 179--193. {IEEE}
  Computer Society, 2003.

\bibitem{beals}
Robert Beals, Harry Buhrman, Richard Cleve, Michele Mosca, and Ronald de~Wolf.
\newblock Quantum lower bounds by polynomials.
\newblock {\em J. {ACM}}, 48(4):778--797, 2001.

\bibitem{beigelsurvey}
Richard Beigel.
\newblock The polynomial method in circuit complexity.
\newblock In {\em Proceedings of the Eigth Annual Structure in Complexity
  Theory Conference, San Diego, CA, USA, May 18-21, 1993}, pages 82--95. {IEEE}
  Computer Society, 1993.

\bibitem{beigel}
Richard Beigel.
\newblock Perceptrons, pp, and the polynomial hierarchy.
\newblock {\em Computational Complexity}, 4:339--349, 1994.

\bibitem{belovsrosmanis}
A.~{Belovs} and A.~{Rosmanis}.
\newblock {Adversary Lower Bounds for the Collision and the Set Equality
  Problems}.
\newblock {\em ArXiv e-prints}, October 2013.

\bibitem{belovsspalek}
Aleksandrs Belovs and Robert Spalek.
\newblock Adversary lower bound for the k-sum problem.
\newblock In Robert~D. Kleinberg, editor, {\em Innovations in Theoretical
  Computer Science, {ITCS} '13, Berkeley, CA, USA, January 9-12, 2013}, pages
  323--328. {ACM}, 2013.

\bibitem{brassard}
G.~Brassard, P.~H\o yer, and A.~Tapp.
\newblock Quantum algorithm for the collision problem.
\newblock {\em ACM SIGACT News (Cryptology Column)}, 28:14--19, 1997.
\newblock quant-ph/9705002.

\bibitem{bvdw}
Harry Buhrman, Nikolai~K. Vereshchagin, and Ronald de~Wolf.
\newblock On computation and communication with small bias.
\newblock In {\em 22nd Annual {IEEE} Conference on Computational Complexity
  {(CCC} 2007), 13-16 June 2007, San Diego, California, {USA}}, pages 24--32.
  {IEEE} Computer Society, 2007.

\bibitem{bt13}
Mark Bun and Justin Thaler.
\newblock Dual lower bounds for approximate degree and markov-bernstein
  inequalities.
\newblock In Fedor~V. Fomin, Rusins Freivalds, Marta~Z. Kwiatkowska, and David
  Peleg, editors, {\em ICALP (1)}, volume 7965 of {\em Lecture Notes in
  Computer Science}, pages 303--314. Springer, 2013.

\bibitem{bt14}
Mark Bun and Justin Thaler.
\newblock Hardness amplification and the approximate degree of constant-depth
  circuits.
\newblock {\em Electronic Colloquium on Computational Complexity {(ECCC)}},
  20:151, 2013.

\bibitem{ctuw}
Karthekeyan Chandrasekaran, Justin Thaler, Jonathan Ullman, and Andrew Wan.
\newblock Faster private release of marginals on small databases.
\newblock {\em CoRR}, abs/1304.3754, 2013.

\bibitem{adversary1}
Peter Hoyer, Troy Lee, and Robert Spalek.
\newblock Negative weights make adversaries stronger.
\newblock In {\em Proceedings of the Thirty-ninth Annual ACM Symposium on
  Theory of Computing}, STOC '07, pages 526--535, New York, NY, USA, 2007. ACM.

\bibitem{rootn}
Peter H{\o}yer, Michele Mosca, and Ronald de~Wolf.
\newblock Quantum search on bounded-error inputs.
\newblock In Jos C.~M. Baeten, Jan~Karel Lenstra, Joachim Parrow, and
  Gerhard~J. Woeginger, editors, {\em Automata, Languages and Programming, 30th
  International Colloquium, {ICALP} 2003, Eindhoven, The Netherlands, June 30 -
  July 4, 2003. Proceedings}, volume 2719 of {\em Lecture Notes in Computer
  Science}, pages 291--299. Springer, 2003.

\bibitem{agnostic}
Adam~Tauman Kalai, Adam~R. Klivans, Yishay Mansour, and Rocco~A. Servedio.
\newblock Agnostically learning halfspaces.
\newblock {\em {SIAM} J. Comput.}, 37(6):1777--1805, 2008.

\bibitem{reliable}
Varun Kanade and Justin Thaler.
\newblock Distribution-independent reliable learning.
\newblock In Maria{-}Florina Balcan and Csaba Szepesv{\'{a}}ri, editors, {\em
  Proceedings of The 27th Conference on Learning Theory, {COLT} 2014,
  Barcelona, Spain, June 13-15, 2014}, volume~35 of {\em {JMLR} Proceedings},
  pages 3--24. JMLR.org, 2014.

\bibitem{ksdnf}
Adam~R. Klivans and Rocco~A. Servedio.
\newblock Learning {DNF} in time
  2\({}^{\mbox{{\~{o}}(n\({}^{\mbox{1/3}}\))}}\).
\newblock {\em J. Comput. Syst. Sci.}, 68(2):303--318, 2004.

\bibitem{klivansservedioomb}
Adam~R. Klivans and Rocco~A. Servedio.
\newblock Toward attribute efficient learning of decision lists and parities.
\newblock {\em Journal of Machine Learning Research}, 7:587--602, 2006.

\bibitem{kutin}
Samuel Kutin.
\newblock Quantum lower bound for the collision problem with small range.
\newblock {\em Theory of Computing}, 1(1):29--36, 2005.

\bibitem{lecturenotes}
Scott~Aaronson lecturer.
\newblock The collision problem: Notes for lecture 13 of mit course 6.845:
  Quantum complexity theory, 2010.

\bibitem{nisanszegedy}
Noam Nisan and Mario Szegedy.
\newblock On the degree of boolean functions as real polynomials.
\newblock {\em Computational Complexity}, 4:301--313, 1994.

\bibitem{odonnellservedio}
Ryan O'Donnell and Rocco~A. Servedio.
\newblock New degree bounds for polynomial threshold functions.
\newblock {\em Combinatorica}, 30(3):327--358, 2010.

\bibitem{paturi}
Ramamohan Paturi.
\newblock On the degree of polynomials that approximate symmetric boolean
  functions (preliminary version).
\newblock In S.~Rao Kosaraju, Mike Fellows, Avi Wigderson, and John~A. Ellis,
  editors, {\em Proceedings of the 24th Annual {ACM} Symposium on Theory of
  Computing, May 4-6, 1992, Victoria, British Columbia, Canada}, pages
  468--474. {ACM}, 1992.

\bibitem{stt}
Rocco~A. Servedio, Li-Yang Tan, and Justin Thaler.
\newblock Attribute-efficient learning and weight-degree tradeoffs for
  polynomial threshold functions.
\newblock In Shie Mannor, Nathan Srebro, and Robert~C. Williamson, editors,
  {\em COLT}, volume~23 of {\em JMLR Proceedings}, pages 14.1--14.19. JMLR.org,
  2012.

\bibitem{sherstov14}
A.~A. Sherstov.
\newblock Breaking the {M}insky-{P}apert barrier for constant-depth circuits.
\newblock In {\em STOC}, 2014.

\bibitem{sherstovsurvey}
Alexander~A. Sherstov.
\newblock Communication lower bounds using dual polynomials.
\newblock {\em Bulletin of the {EATCS}}, 95:59--93, 2008.

\bibitem{sherstovmajmaj}
Alexander~A. Sherstov.
\newblock Separating ac\({}^{\mbox{0}}\) from depth-2 majority circuits.
\newblock {\em {SIAM} J. Comput.}, 38(6):2113--2129, 2009.

\bibitem{patternmatrix}
Alexander~A. Sherstov.
\newblock The pattern matrix method.
\newblock {\em {SIAM} J. Comput.}, 40(6):1969--2000, 2011.

\bibitem{comm5}
Alexander~A. Sherstov.
\newblock Strong direct product theorems for quantum communication and query
  complexity.
\newblock In Lance Fortnow and Salil~P. Vadhan, editors, {\em Proceedings of
  the 43rd {ACM} Symposium on Theory of Computing, {STOC} 2011, San Jose, CA,
  USA, 6-8 June 2011}, pages 41--50. {ACM}, 2011.

\bibitem{sherstov13}
Alexander~A. Sherstov.
\newblock Approximating the and-or tree.
\newblock {\em Theory of Computing}, 9(20):653--663, 2013.

\bibitem{sherstovfocs}
Alexander~A. Sherstov.
\newblock The intersection of two halfspaces has high threshold degree.
\newblock {\em {SIAM} J. Comput.}, 42(6):2329--2374, 2013.

\bibitem{shi}
Yaoyun Shi.
\newblock Quantum lower bounds for the collision and the element distinctness
  problems.
\newblock In {\em 43rd Symposium on Foundations of Computer Science {(FOCS}
  2002), 16-19 November 2002, Vancouver, BC, Canada, Proceedings}, pages
  513--519. {IEEE} Computer Society, 2002.

\bibitem{shizhu}
Yaoyun Shi and Yufan Zhu.
\newblock Quantum communication complexity of block-composed functions.
\newblock {\em Quantum Information {\&} Computation}, 9(5):444--460, 2009.

\bibitem{spalek}
Robert Spalek.
\newblock A dual polynomial for {OR}.
\newblock {\em CoRR}, abs/0803.4516, 2008.

\bibitem{adversary3}
Robert Spalek and Mario Szegedy.
\newblock All quantum adversary methods are equivalent.
\newblock {\em Theory of Computing}, 2(1):1--18, 2006.

\bibitem{thaler}
Justin Thaler.
\newblock Lower bounds for the approximate degree of block-composed functions.
\newblock {\em Electronic Colloquium on Computational Complexity {(ECCC)}},
  22:150, 2014.

\bibitem{tuv}
Justin Thaler, Jonathan Ullman, and Salil~P. Vadhan.
\newblock Faster algorithms for privately releasing marginals.
\newblock In Artur Czumaj, Kurt Mehlhorn, Andrew~M. Pitts, and Roger
  Wattenhofer, editors, {\em Automata, Languages, and Programming - 39th
  International Colloquium, {ICALP} 2012, Warwick, UK, July 9-13, 2012,
  Proceedings, Part {I}}, volume 7391 of {\em Lecture Notes in Computer
  Science}, pages 810--821. Springer, 2012.

\bibitem{yuen}
Henry Yuen.
\newblock A quantum lower bound for distinguishing random functions from random
  permutations.
\newblock {\em Quantum Information {\&} Computation}, 14(13-14):1089--1097,
  2014.

\bibitem{zhandry}
Mark Zhandry.
\newblock A note on the quantum collision and set equality problems.
\newblock {\em Quantum Information {\&} Computation}, 15(7{\&}8):557--567,
  2015.

\bibitem{adversary4}
Shengyu Zhang.
\newblock On the power of ambainis lower bounds.
\newblock {\em Theoretical Computer Science}, 339(2–3):241 -- 256, 2005.

\end{thebibliography}

\begin{appendix}
\section{On the Tightness of Theorem \ref{thm:collision_main} and Corollary \ref{cor:ed_lower}} \label{app:tightness}

To complement Theorem \ref{thm:collision_main}, we construct an approximating polynomial that gives a nearly matching upper bound on the approximate degree of $\col_{N, R}$. The construction is a refinement of the approximating polynomial given in Section \ref{sec:complementaryslackness}.

\begin{proposition}
For $0 \le \delta \le 1/N$, there exists a polynomial $p$ of degree $O(\delta^{1/3} N^{1/3} \log R)$ that $(1-\delta)$-approximates $\col_{N, R}$.
\end{proposition}

\begin{proof}[Proof sketch]
See Section \ref{sec:complementaryslackness} for the construction of an approximating polynomial of degree $O(N^{1/3} \log R)$ in the case where $\delta$ is constant. In order to obtain an improved upper bound for vanishing $\delta$, we make the following changes to that construction:
\begin{enumerate}
\item We instead choose $r = \delta^{1/3}N^{1/3}$. Now if $x$ is a 2-1 input, the probability over the random choice of the set $S$ of obtaining a collision inside $S$, i.e. the probability that $\mathbb{I}_{\ED, S} = 0$, is at most $(N/2) \cdot (r/N)^2 \le \delta / 2$.
\item We instead let $A_d$ be an affine transformation of a Chebyshev polynomial with the following properties for some constant $M$:

\begin{itemize}
\item $A_d(0) \ge \frac{\delta}{2}$
\item $A_d(i) \in [-1, -\frac{\delta}{2}]$ for $i \in[1, d^2/M\delta]$
\item $A_d(i) \in [-1, 1]$ for $x \in [0, d^2/M\delta]$. 
\end{itemize}

\item Setting $d = 100 \cdot M \cdot r$ ensures that the polynomial $p$ has degree $O(\delta^{1/3} N^{1/3} \log R)$ and is a $(1-\delta)$-approximation of $\col_{N, R}$.
\end{enumerate}

\end{proof}

We now show that Corollary \ref{cor:ed_lower} is tight up to a factor of $\log R$, when $\delta \le 1/M^2$. This gives mild evidence that the lower bound has the right dependence on both parameters $M, \delta$ for vanishing $\delta$.

\begin{proposition}
Let $\delta \le 1/M^2$. Then there exists a $(1-\delta)$-approximating polynomial for $\ED_{M, R}$ with degree $O(\log R)$.
\end{proposition}

\begin{proof}
We write
\[\ED_{M, R}(x_1, \dots, x_M) =  \bigwedge_{i \ne j} \NEQ(x_i, x_j),\]
where $\NEQ(x_i, x_j) = 1$ if $i$ and inputs $j$ are distinct, and is zero otherwise. The function $\NEQ$ can be computed exactly by a polynomial of degree $O(\log R)$. Therefore, the polynomial
\[\frac{1}{{M \choose 2}} \left( \frac{1}{2} - \sum_{i \ne j} \NEQ(x_i, x_j) \right )\]
has degree $O(\log R)$ and approximates $\ED_{M, R}$ to within error $1-1/M^2$.
\end{proof}

\end{appendix}

\end{document}